\documentclass[copyright,creativecommons]{eptcs}
\providecommand{\event}{MEMICS 2016}
\usepackage{url}	% allows at least the use of urls in references
\usepackage{amsmath}
\usepackage{amsthm}
\usepackage{amssymb}
\usepackage{pgf}
\usepackage{tikz}
\usepackage{underscore}           % Only needed if you use pdflatex. BUT YOU MUST REMOVE UNDERSCORES FROM ALL THE NAMES OF FILES USED.
\usepackage[english]{babel}	  % If you use the \usepackage{underscore} the '_' character in labels will cause this error. To fix this also use \usepackage[english]{babel}
\usepackage{graphicx,subfig}
\usepackage{caption}

\title{Reducing Nondeterministic Tree Automata by \\Adding Transitions}
\def\titlerunning{Reducing Nondeterministic Tree Automata by Adding Transitions}
\def\authorrunning{Ricardo Almeida}
\author{Ricardo Manuel de Oliveira Almeida
\institute{University of Edinburgh, United Kingdom}
}

\newcommand{\mytab}{\quad}

\newtheorem{definition}{Definition}[section]
\newtheorem{theorem}{Theorem}
\newtheorem{lemma}{Lemma}[section]

\newcommand{\A}{A}

\newcommand{\allcompletetrees}{\mathbb C}    
\newcommand{\alltrees}{\mathbb T}

\newcommand{\dom}{dom}
\newcommand{\drel}{R_\mathsf{d}}		% Some downward relation
\newcommand{\ldrel}{\hat{R}_\mathsf{d}}		% Some downward relation

\newcommand{\dwsimlarg}{\xsimlarg{\mathsf{dw}}}
\newcommand{\dwsim}{\xsim{\mathsf{dw}}}
\newcommand{\dwequiv}{\equiv^\mathsf{dw}}
\newcommand{\upequiv}{\equiv^\mathsf{up}}
\newcommand{\dwtraceinclusionlarg}{\xtraceinclusionlarg{\mathsf{dw}}}
\newcommand{\dwtraceinclusion}{\xincl{\mathsf{dw}}}

\newcommand{\f}{a}

\newcommand{\id}{{\it id}} % Identity relation
\newcommand{\transasymkupsim}[2]{\prec^{#1\textrm{-}\mathsf{up}}\!\!(#2)}

\newcommand{\transkdwsim}[1]{\preceq^{#1\textrm{-}\mathsf{dw}}}
\newcommand{\transkdwsimlarg}[1]{\succeq^{#1\textrm{-}\mathsf{dw}}}
\newcommand{\transkupsim}[2]{\preceq^{#1\textrm{-}\mathsf{up}}\!\!(#2)}
\newcommand{\transkupsimlarg}[2]{\succeq^{#1\textrm{-}\mathsf{up}}\!\!(#2)}
\newcommand{\transasymkdwsim}[1]{\prec^{#1\textrm{-}\mathsf{dw}}}

\newcommand{\transTD}[4]{\langle{#1},{#2},\vect{#3}{#4}\rangle}	% Top-down transition.
\newcommand{\nat}{\mathbb{N}}

\newcommand{\node}{v}

\newcommand{\rank}{\#}
\newcommand{\run}{\pi} 
\newcommand{\spstate}{\psi}						% The special state in our definition of TA.

\newcommand{\strictdwsim}{\strictxsim{\mathsf{dw}}}

\newcommand{\strictupsim}[1]{\strictxsim{\mathsf{up}}\!\!\!\;\!(#1)}

\newcommand{\strictxsim}[1]{\sqsubset^{#1}}

\newcommand{\tickNO}{\times}
\newcommand{\tickOK}{\checkmark}  
\newcommand{\tree}{t}
\newcommand{\true}{\mathsf{TRUE}}

\newcommand{\upsim}[1]{\xsim{\mathsf{up}}\!\!\!\;\!\!(#1)}
\newcommand{\upsimlarg}[1]{\xsimlarg{\mathsf{up}}\!\!(#1)}
\newcommand{\uptraceinclusion}[1]{\xincl{\mathsf{up}}\!\!\!\;\!(#1)}
\newcommand{\uptraceinclusionlarg}[1]{\xincllarg{\mathsf{up}}\!\!\!\;\!(#1)}
\newcommand{\urel}{R_\mathsf{u}}		% Some upward relation
\newcommand{\vect}[2]{#1_1\dotsc#1_{#2}}
\newcommand{\xincl}[1]{\subseteq^{#1}}
\newcommand{\xincllarg}[1]{\supseteq^{#1}}

\newcommand{\xsim}[1]{\sqsubseteq^{#1}}
\newcommand{\xsimlarg}[1]{\sqsupseteq^{#1}}

\newcommand{\xtraceinclusionlarg}[1]{\supseteq^{#1}}

\newcommand{\pspace}{\texttt{PSPACE}}
\newcommand{\exptime}{\texttt{EXPTIME}}

\newcommand{\rabit}{\texttt{Rabit}}
\newcommand{\libvata}{\texttt{libvata}}
\newcommand{\minotaut}{\texttt{minotaut}}

\newcommand{\level}{\mathit{level}}

\newcommand{\ignore}[1]{}
\newcommand{\ru}{{\it RU}}

\begin{document}
\maketitle

\begin{abstract}
We introduce saturation of nondeterministic tree automata, 
a technique that adds new transitions to an automaton while preserving its language. 
We implemented our algorithm on \minotaut\ 
- a module of the tree automata library \libvata\ that reduces the size of automata by merging states and removing superfluous transitions - 
and we show how saturation can make subsequent merge and transition-removal operations more effective. 
Thus we obtain a Ptime algorithm that reduces the size of tree automata even more than before. 
Additionally, 
we explore how \minotaut\ alone can play an important role when performing hard operations like complementation, 
allowing to both obtain smaller complement automata and lower computation times. 
We then show how saturation can extend this contribution even further. 
We tested our algorithms on a large collection of automata from applications of \libvata\ in shape analysis, 
and on different classes of randomly generated automata.
\end{abstract}

\section{Introduction}

Tree automata are a generalization of word automata to non-linear words (i.e., trees)~\cite{tata2008}.
They have many applications in 
model checking~\cite{abdulla:tree,Bouajjani:modelChecking2006}, 
% model checking~\cite{abdulla:tree,Abdulla:transducers05,Bouajjani:modelChecking2006}, 
term rewriting~\cite{tool:Autowrite}
and related areas of formal software verification, e.g., 
shape analysis \cite{Holik:shapeForest2013}.
Several software packages for manipulating tree automata have been developed, e.g.,
Timbuk \cite{tool:Timbuk}, Autowrite
\cite{tool:Autowrite}
and \libvata\ \cite{tool:libvata} 
(on which other verification tools, like Forester \cite{tool:forester}, are based).

For nondeterministic automata,
many questions about their languages are computationally hard. 
The language universality, equivalence and inclusion
problems are \pspace-complete for word automata and \exptime-complete for tree automata~\cite{tata2008}.
A common approach to solving many instances of the inclusion problem is via the computation
of different notions of simulation preorders that at the same time under-approximate language inclusion
and are computable in polynomial time~\cite{etessami:hierarchy2002,Holik:computingSim08}.
These simulation preorders thus offer a trade-off between computability and expressiveness.
Efficient reduction algorithms have been presented both for word automata~\cite{mayr:advanced2013} 
and for tree automata~\cite{lukas:framework2009,Almeida2016},
where language inclusion is witnessed by the membership of a pair of states in a simulation preorder.
In our paper, 
we focus on Heavy(x,y)~\cite{Almeida2016},
a polynomial-time algorithm for reducing tree automata,
in the sense of obtaining a smaller automaton with the same language,
though not necessarily with the absolute minimal number of states possible
(in general, as with word automata,
there is no unique nondeterministic automaton with the minimal possible number of states for a given language).
Heavy(x,y) is based on an intricate combination of transition pruning and state quotienting techniques for tree automata,
extending previous work on the words case~\cite{mayr:advanced2013}.
Transition pruning is based on the notion that certain transitions may be removed from the automaton because 'better'
ones remain.
The notion of 'better' is given by comparing the states at the endpoints of the two transitions w.r.t.\ suitable 
simulation preorders. 
The Heavy(x,y) algorithm yields substantially smaller and sparser 
(i.e., using fewer transitions per state and per symbol) automata than all previously known
reduction techniques,
and it is still fast enough to handle large instances.

We start by optimizing the computation of simulation preorders in Heavy(x,y).
This is done by identifying re-computations that can be skipped,
which yields generally faster computation times.
We then introduce the dual notion of transition pruning, 
in which transitions are added to the automaton if 'better' ones exist already.
This technique is known as transition saturation and it was previously defined for word 
automata~\cite{lorenzomayr:reductionAut2006}.
As in transition pruning,
this technique compares the source states of the two transitions w.r.t.\ a simulation $R_s$ on the states space,
and the target states of the transitions w.r.t.\ a simulation $R_t$.
If saturating an automaton with $R_s$ and $R_t$ preserves the language,
we say that $S(R_s,R_t)$ is good for saturation.
We provide a summary of all $S(R_s,R_t)$ we found to be or not to be good for saturation.

The motivation behind saturation is that it may allow for new merging of states and transition removal which 
were not possible by using Heavy alone.
Thus saturating an automaton which has been reduced with Heavy(x,y) and then reducing it again
might result in an even smaller automaton.
We perform an experimental evaluation to measure how much smaller,
on average,
automata become by interleaving reduction methods with transition saturation.
Our results indicate that generally one obtains automata with fewer states,
but on some cases with more transitions, than the ones obtained by Heavy(x,y) alone.

In general, 
one wishes to reduce automata in order to make them more efficient to handle in subsequent computations.
Thus, we present a second experimental evaluation showing that the complement automata are much smaller and 
faster to compute when 
the automata have previously been reduced with the techniques described above.

We implemented our algorithm as an extension of \minotaut\ (source code available \cite{tool:minotaut}), 
a module of the tree automata library \libvata~\cite{tool:libvata} where the Heavy algorithm is provided.
The experiments described above were performed on a large collection of automata 
from applications of \libvata\ in shape analysis, 
as well as on different classes of randomly generated tree automata.

\section{Preliminaries} \label{sec:preliminaires}

\paragraph{Trees and tree automata.} 

A \emph{ranked alphabet} $\Sigma$ is a set
of symbols together with a function $\rank : \Sigma \rightarrow \nat_0$. 
For $\sigma \in \Sigma$, $\rank{(\sigma)}$ is called the \emph{rank} of $\f$. 
We define a \emph{node} as a sequence in $\nat^*$.
For a node $\node \in \nat^*$, we define the $i$-th child of $v$ to be the node $vi$, for some $i \in \nat$.

Given a ranked alphabet $\Sigma$, 
a finite \emph{tree} over $\Sigma$ is defined as a partial mapping $\tree: \nat^* \rightarrow \Sigma$ such that 
for all $v \in \nat^*$ and $i \in \nat$, 
if $vi \in \dom(\tree)$ then 
\textbf{(1)} $\node\in\dom(\tree)$, and
\textbf{(2)} $\rank(\tree(\node))\geq i$.
Note that the number of children of a node $v$ 
may be smaller than $\#(t(v))$.
In this case we say that the node is \emph{open}. 
Nodes which have exactly $\#(t(v))$ children are called \emph{closed}.
Nodes which do not have any children are called \emph{leaves}.
A tree is closed if all its nodes are closed, otherwise it is open.
By $\allcompletetrees(\Sigma)$ we denote the set of all closed trees over $\Sigma$
and by $\alltrees(\Sigma)$ the set of all trees over $\Sigma$.
% A tree $\tree'$ is a prefix of $\tree$ iff $\dom(\tree')\subseteq\dom(\tree)$ and 
% for all $\node\in\dom(\tree')$, $\tree'(\node) = \tree(\node)$.

A finite nondeterministic \emph{top-down tree automaton} (TDTA) is a quadruple 
$\A = (\Sigma,Q,\delta, I)$ where $Q$ is a finite set of states, $I \subseteq Q$ is a
set of initial states, $\Sigma$ is a ranked alphabet, 
and $\delta \subseteq Q \times \Sigma\times Q^+ $ is the transition relation. 
A TDTA has an unique final state, which we represent by $\spstate$.
The transition rules satisfy that if $\langle q, \sigma, \spstate \rangle \in\delta$ then $\rank(\sigma) = 0$, 
and if $\langle q, \sigma, q_1 \ldots q_n \rangle \in \delta$ (with $n > 0$) then $\rank(\sigma) = n$.
Informally, a run of $\A$ reads an input tree top-down from the root, 
branching into sub-runs on subtrees as specified by the applied transition rules,
and it accepts it if every branch ends in $\spstate$.
Formally, a run of $\A$ over a tree $\tree \in \alltrees(\Sigma)$ (or a $\tree$-run in $\A$) 
is a partial mapping $\run : \nat^*\rightarrow Q$ such that $\node\in\dom(\run)$ iff 
either $\node\in\dom(\tree)$ or $\node = \node'i$ where $\node'\in\dom(\tree)$ and $i \leq \rank(\tree(\node'))$.  
Further, for every $\node\in\dom(\tree)$, there exists either
\textbf{a)} a rule $\langle q, a, \spstate \rangle$ such that $q = \run(\node)$ and $\sigma = \tree(\node)$,
or \textbf{b)} a rule $\transTD q \sigma q n$ such that $q = \run(\node)$, 
$\sigma = \tree(\node)$, and $q_i = \run(\node i)$ for each $i:1\leq i\leq \rank(\sigma)$.
A \emph{leaf of a run} $\run$ on $\tree$ is a node $\node\in\dom(\run)$ 
such that $\node i\in\dom(\run)$ for no $i \in\nat$.
% We call it \emph{dangling} if $\node\not\in\dom(\tree)$. 
% Intuitively, the dangling nodes of a run over $t$ are all the nodes which are in $\run$ but are missing 
% in $t$ due to it being incomplete.
% Notice that dangling leaves of $\pi$ are children of open nodes of $\tree$. 
% The prefix of depth $k$ of a run $\run$ is denoted $\run_k$. 
% Runs are always finite since the trees we are considering are finite.

We write $\tree\stackrel{\run}{\Longrightarrow} q$ to denote
that $\run$ is a $t$-run of $\A$ such that $\run(\epsilon)=q$. 
A run $\run$ is accepting if $\tree\stackrel{\run}{\Longrightarrow} q \in I$.
The \emph{downward language of a state} $q$ in $\A$ is defined by
$D_A(q)=\{\tree\in \allcompletetrees(\Sigma) \mid \tree \stackrel{\run}{\Longrightarrow} q, \mbox{for some run } \run \}$, 
while the \emph{language} of $\A$ is defined by $L(\A)=\bigcup_{q\in I}D_A(q)$. 
We sometimes write simply $\A$ to refer to its language.
% The \emph{upward language} of a state $q$ in $\A$, denoted $\upl_A(q)$, is then defined as the set of open trees $\tree$,
% such that there exists an accepting $\tree$-run $\run$ with exactly one dangling leaf $\node$ s.t. 
% $\run(\node) = q$.
% We omit the $\A$ subscript notation when it is implicit which automaton we are considering.

\paragraph{Downward and upward relations.} 

The behaviour of states in TDTA can be compared by semantic
preorders (and their induced equivalences), based on the upward- or
downward behaviour of the automaton from these states.

Ordinary downward simulation on tree automata 
can be characterized by a game between two players, Spoiler and Duplicator. 
Given a pair of states $(q,r)$, 
Spoiler wants to show that $(q,r)$ is not contained in the simulation preorder relation, 
while Duplicator has the opposite goal.
Starting in the initial configuration $(q,r)$, 
Spoiler chooses a transition $q \stackrel{\sigma}{\longrightarrow} \langle q_1 \ldots q_n \rangle$,
where $n = \#(\sigma)$,
and Duplicator must imitate it \emph{stepwise} by 
choosing a transition with the same symbol $r \stackrel{\sigma}{\longrightarrow} \langle r_1 \ldots r_n \rangle$. 
This yields $n$ new configurations $(q_1,r_1), \ldots, (q_n,r_n)$ from which the game continues independently.
If a player ever cannot make a move then the other player wins. 
Duplicator wins every infinite game. Simulation holds iff Duplicator wins. 

A tree branches as one goes downward, but `joins in' side branches as one goes upward.
Therefore a comparison of the upward behaviour of states depends also
on the joining side branches as one goes upward in the tree.
Thus upward simulation is only defined {\em relative} to a given
other relation $R$ that compares the downward behaviour of states `joining in' from
the sides \cite{Holik:computingSim08}. 
One speaks, e.g., of upward simulation {\em of} $R$.
Thus in the ordinary upward simulation game, starting in the initial configuration $(q,r)$, 
Spoiler chooses a transition $q' \stackrel{\sigma}{\longrightarrow} \langle q_1 \ldots q_n \rangle$,
where $q=q_i$ for some $i$ and $n = \#(\sigma)$,
and Duplicator must imitate it \emph{stepwise} by 
choosing a transition with the same symbol $r' \stackrel{\sigma}{\longrightarrow} \langle r_1 \ldots r_n \rangle$,
where $r=r_i$, and such that 
1) $q_j R r_j$, for every $j \neq i$, and
2) $q \in I \Longrightarrow r \in I $.
The game continues from the configuration $(q',r')$,
and Spoiler wins if Duplicator ever cannot respond to a move,
otherwise Duplicator wins.

While in ordinary downward simulation (resp., upward simulation w.r.t.\ $R$) 
Duplicator only knows Spoiler's very next step,
in downward $k$-lookahead simulation (resp., upward $k$-lookahead simulation w.r.t.\ $R$) 
Duplicator knows Spoiler's next $k$ steps in advance 
(unless Spoiler's move ends in a deadlocked state - i.e., a state with no transitions).
In the case where Duplicator knows {\em all} steps of Spoiler in the entire downward simulation game in advance 
(i.e., $k=\infty$),
we talk of downward trace/language inclusion (resp., upward trace inclusion w.r.t.\ $R$).
As the parameter $k$ increases, 
the $k$-lookahead simulation relation becomes larger and thus approximates the respective 
trace inclusion relation better and better.

The downward/upward $k$-lookahead simulation preorder 
(denoted $\transkdwsim{k}$/$\transkupsim{k}{R}$, or just $\dwsim$/$\upsim{R}$ in the ordinary case) 
is the set of all pairs $(p,q)$ for which Duplicator has a winning strategy in the respective game.
For the downward/upward trace inclusion preorder we write $\dwtraceinclusion$/$\uptraceinclusion{R}$.

Downward/upward $k$-lookahead simulation is PTIME-computable for every fixed $k$ and a good
under-approximation of the respective trace inclusion
(which is EXPTIME-complete in the downward case \cite{tata2008}, 
and PSPACE-complete for $R$ = $\id$ in the upward case).

\paragraph{Transition pruning and state quotienting.} 

Given a TDTA $\A = (\Sigma,Q,\delta, I)$,
certain transitions may be pruned without changing the language, 
because `better' ones remain.
Given a strict partial order $P \subseteq \delta \times\delta$ on the set of transitions, 
the pruned automaton is defined as ${\it Prune}(\A,P) = (\Sigma,Q,\delta', I)$ where
$\delta' = \{ (p,\sigma,r) \in \delta \mid \nexists (p',\sigma,r') \in \delta \ldotp (p,\sigma,r) \, P \, (p',\sigma,r') \} $.
I.e., if $t\, P \,t'$ then $t$ may be pruned because $t'$ is `better' than $t$. 
${\it Prune}(A,P)$ is unique and transitions are removed in parallel without re-computing $P$. 
% (which could be expensive).
Trivially, $L(Prune(\A,P)) \subseteq L(\A)$.
If $L({\it Prune}(A,P)) = L(A)$ also holds we say that $P$ is \emph{good for pruning} (GFP).

We obtain GFP relations by comparing the endpoints of transitions over the same symbol $\sigma \in \Sigma$.
Given two binary relations $\urel$ and $\drel$ on $Q$, we define 
$P(\urel,\drel) = \{ (\langle p, \sigma, r_1 \cdots r_n \rangle , \langle p', \sigma, r'_1 \cdots r'_n \rangle) \, \mid \, p \mathrel{\urel} p' \mbox{ and } (r_1 \cdots r_n) \mathrel{\ldrel} (r'_1 \cdots r'_n) \},$
where $\ldrel$ is a suitable lifting of $\drel \subseteq Q \times Q$ to $\ldrel \subseteq Q^n \times Q^n$:
if $\drel$ is some strict partial order $<_\mathsf{d}$,
then $\ldrel$ is a binary relation $\hat{<}_\mathsf{d}$ s.t.
1) $\forall_{1 \leq i \leq n} \ldotp r_i \leq_\mathsf{d} r'_i$, and
2) $\exists_{1 \leq i \leq n} \ldotp r_i <_\mathsf{d} r'_i$;
if $\drel$ is a non-strict partial order $\leq_\mathsf{d}$, then only condition 1) applies.
The relations $\urel,\drel$ are chosen such that 
$P(\urel,\drel) \subseteq \delta \times \delta$ 
is a strict partial order
(i.e., of the two relations $\urel$ and $\drel$, one must be a strict partial order)
that is GFP; see the algorithm Heavy below.

Another method for reducing the size of automata is state quotienting. 
Given a suitable equivalence on the set of states, 
each equivalence class is collapsed into just one state.
A preorder $\sqsubseteq$ induces an equivalence relation 
$\equiv \;:=\; \sqsubseteq \!\cap\! \sqsupseteq$.
Given $q \in Q$, $[q]$ denotes its equivalence class w.r.t.\ $\equiv$.
For $P \subseteq Q$, $[P]$ denotes the set of equivalence classes 
$[P] = \{[p]\,|\, p \in P\}$.
The quotient automaton is defined as
$\A/\!\equiv \; := \; (\Sigma, [Q], \delta_{A/\!\equiv}, [I])$,
where $\delta_{A/\!\equiv} = \{ \langle [q],\sigma,[q_1]\ldots[q_n] \rangle
\mid \langle q,\sigma,q_1\ldots q_n \rangle \in \delta_A \}$.
Trivially, $L(A) \subseteq L(A/\!\!\equiv)$.
If $L(A) = L(A/\!\!\equiv)$ also holds, 
$\equiv$ is said to be \emph{good for quotienting} (GFQ).

\paragraph{The Heavy algorithm.} 

Here we describe Heavy(x,y)~\cite{Almeida2016},
a tree automata reduction algorithm based on transition pruning and 
state quotienting.
The parameters $x,y \ge 1$ describe the lookahead for the used downward/upward lookahead simulations, respectively,
where larger values yield better reduction but are harder to compute.
The algorithm is polynomial for fixed $x,y$,
and doubly exponential in $x$ (due to the downward branching of the tree) and single exponential in $y$ otherwise.
Let ${\it Op}(x,y)$ be the following sequence of operations on tree automata,
where ${\it RU}$ stands for removing useless states 
(i.e., states that cannot be reached from any initial state or from which no tree can be accepted):
$\ru$, 
quotienting with $\transkdwsim{x}$, 
pruning with $P(\id,\transasymkdwsim{x})$,
$\ru$,
quotienting with $\transkupsim{y}{\id}$,
pruning with $P(\transasymkupsim{y}{\id}, \id)$,
pruning with $P(\strictupsim{\id},\transkdwsim{x})$,
$\ru$,
quotienting with $\transkupsim{y}{\id}$,
pruning with $P(\transkupsim{y}{\dwsim},\strictdwsim)$,
$\ru$.
These operations are language preserving, 
since the used relations are GFP/GFQ \cite{Almeida2016}.
% The order of the operations is chosen for efficiency.

The algorithm Heavy(1,1) just iterates ${\it Op}(1,1)$ 
until a fixpoint is reached.
% For efficiency reasons, 
The general algorithm Heavy(x,y) does not iterate ${\it Op}(x,y)$, 
but uses a double loop: 
it iterates the sequence Heavy(1,1)${\it Op}(x,y)$ until a fixpoint is reached.

The Heavy algorithm is provided in the \minotaut\ library~\cite{tool:minotaut},
making use of \libvata's efficient computation of ordinary simulation
(for a description of \minotaut's implementation of simulation with larger lookaheads see Section~\ref{sec:efficient}).
Heavy behaves well in practice, 
significantly reducing both automata of program verification provenience 
and randomly generated automata~\cite{Almeida2016}.

\section{Efficient Computation of Lookahead Simulations}	\label{sec:efficient}

We performed some optimizations on the computation of the maximal downward lookahead simulation used 
in Heavy(x,y).
In the following we describe the key aspects of the computation in terms of a game between Spoiler and Duplicator.
(Upward simulation is similar but simpler, since the tree branches downward.) 

\paragraph{Fixpoint iteration with incremental moves.} 
We represent binary relations over $Q$ as boolean matrices of dimension
$|Q|\times|Q|$.
Starting with a matrix $W$ in which all entries are set to $\true$,
the algorithm consists of a downward refinement loop of $W$
that converges to the maximal downward $k$-lookahead simulation.
In each iteration of the refinement loop, for each pair $p,q$ where $W[p][q]$ is still $\true$:
\begin{itemize}
 \item Spoiler tries an attack ${\it atk}$ consisting of a possible move from $p$ 
 of some depth $d \le k$.
 Each such attack is built incrementally, 
 for $d=1,2,\dots,k$, 
 in order to give Duplicator a chance to respond already to a prefix of ${\it atk}$ of depth $<k$.
\item 
 Duplicator then attempts to defend against the given attack of depth $d$,
 by finding a matching move ${\it def}$ from $q$ by the same symbols s.t.\
 every leaf-state in ${\it def}$ is in relation $W$ with the corresponding state in
 ${\it atk}$. (Duplicator's search is done in depth-first mode.)
 If successful, 
 Duplicator declares victory against this particular (prefix of an) attack and Spoiler tries a new one, 
 since extending the current one to a higher depth is pointless.
 If unsuccessful and $d<k$, 
 Spoiler builds an attack of the next depth level $d+1$,
 by extending $atk$ with one new transition from each of its leaf-states.
 The extra information might enable Duplicator to find a successful defence then.
\item
 Duplicator fails if he could not defend against an attack $atk$ of the maximal depth, 
 either where $atk$ has depth $d=k$
 or $d < k$ but $atk$ cannot be extended any more due to all its leaf-states having no outgoing transitions.
\item
 If Duplicator could defend against every attack (or some prefix of it) by Spoiler then $W[p][q]$ stays true, for now.
\item 
 In the worst case, for each Spoiler's attack of depth $d$, 
 Duplicator must search through all defences of depth up-to $d$, 
 but often Duplicator wins sooner.
\item
 Similarly, in the worst case, Spoiler needs to try all possible attacks of depth $k$, 
 but often Duplicator already wins against prefixes of some depth $d<k$.
\end{itemize} 
Since the outcome of a local game depends on the values of $W$, 
the refinement loop might converge only after several iterations.
The reached fixpoint represents a relation that is generally not transitive (for $k>1$), 
but its transitive closure is the required maximal downward $k$-lookahead simulation preorder $\transkdwsim{k}$.

\paragraph{An Optimization Based on Pre-Refinement.}
Following an approach implemented in \rabit~\cite{tool:rabiturl} for word automata, 
we under-approximate non-simulation as follows.
If there exists a tree of bounded depth $d$ that can be read from state $p$ but not from state $q$, 
then the pair $(p,q)$ cannot be in $k$-lookahead simulation for any $k$.
The pre-refinement step iterates through all pairs $(p,q)$ and sets $W[p][q]$ to $\texttt{false}$ if such a tree is
found witnessing non-simulation.
Our experiments show that, for most automata samples,
running a pre-refinement with some modest depth $d$ suffices to speed up the 
$k$-lookahead downward simulation computation.

We now present an optimization that allows to
compute lookahead simulation faster.
The idea is that attacks which are \emph{good} (i.e., successful) or \emph{bad} (i.e., unsuccessful)
may be remembered to skip unnecessary re-computations.

\paragraph{Semi-global caching of Spoiler's attacks.}
An attack is seen as \emph{good} or \emph{bad} within the scope of the \emph{whole game}.
Consider the game configuration 
$(p_1,q_1)$ in Figure~\ref{fig:opt_automata}. 
Although $q_1$ can read all trees of depth $3$ that $p_1$ can read,
there are \emph{good} attacks from $p_2$ both against $q_2$ and against $q_3$.
Duplicator will find and store these if, when defending against the attack $ac(e,e)$,
he first tries the transition to $q_2$ (which can only read $d$), 
or when defending against $ad(e,e)$ he first tries the transition to $q_3$
(which can only read $c$).
After trying possibly all attempts,
Duplicator is able to defend against the attack and Spoiler now tries the $b$-transition 
from $p_1$ to $p_2$.
However, 
all possible sub-attacks are now the same,
which makes Duplicator announce defeat on them immediately without any exploration.

In Appendix~\ref{appendix:caching} two different ways of performing this caching of Spoiler's attacks can be found.
The three versions present a trade-off between expressiveness and space required to encode attacks.
Our tests indicate that the semi-global version indeed speeds up the computation on automata 
with high transition overlaps (i.e., where many states are shared by different transitions).

\begin{figure}[htbp]
 \begin{center}
%  \quad
 \begin{Large}
  \scalebox{.50}{\input{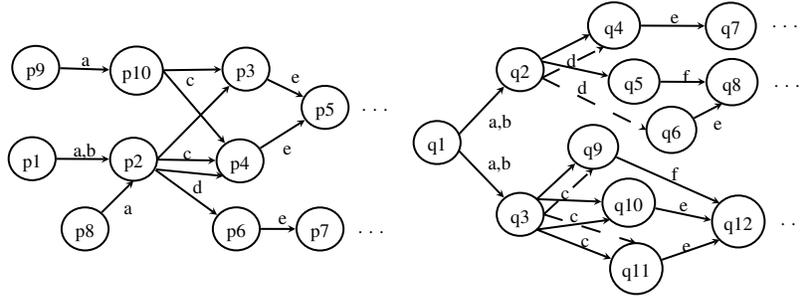}} 
 \end{Large}
 \end{center}
%  \vspace{-10pt}
 \caption{For $W = \{(p_5,q_7),(p_5,q_8),(p_7,q_8),(p_7,q_{12})\}$, 
 all versions of the optimization allow some attacks to be skipped when computing the $3$-lookahead downward simulation.}
 \label{fig:opt_automata}
\end{figure}

\section{Saturation of Tree Automata}

In Section~\ref{sec:preliminaires}, 
we described the transition pruning technique,
which removes a transition if a 'better' one remains.
In this section,
we introduce its dual notion, saturation,
which adds a transition if a 'better' one exists already.
The motivation behind saturation is to pave the way for further reductions
when the Heavy algorithm has reached a fixpoint on the automaton
(see Section~\ref{sec:results}).
Saturation has been defined for the words case before~\cite{lorenzomayr:reductionAut2006},
here we apply it to tree automata.

\vspace{0.2cm}
\begin{definition}
 Let $\A = (\Sigma,Q,\delta,I)$ be a TDTA, $\Delta = Q \times \Sigma \times Q^+$ 
 and $S \subseteq \Delta \times \Delta$ a reflexive binary relation on $\Delta$.
 The \emph{S-saturated automaton} is defined as $Sat(\A,S) := (\Sigma,Q,\delta_S,I)$, where
%  \begin{small}
 $$ \delta_S = \{ \langle p',a,q_1' \ldots q_{\#(a)}' \rangle \!\in\! \Delta \mid \exists \langle p,a,q_1 \ldots q_{\#(a)} \rangle \!\in\! \delta \,\cdot\, \langle p',a,q_1' \ldots q_{\#(a)}' \rangle \,S\, \langle p,a,q_1 \ldots q_{\#(a)} \rangle \}. $$
%  \end{small}
\end{definition}

Since $S$ is reflexive, any transition in the initial automaton is preserved and so 
$\A \subseteq Sat(\A,S)$.
When the converse inclusion also holds,
we say that $S$ is \emph{good for saturation} (GFS).
Note that the GFS property is downward closed in the space of reflexive relations, i.e.,
if $R$ is GFS and $\id \!\subseteq\! R' \!\subseteq\! R$,
then $R'$ too is GFS. 
(or if $R'$ is not GFS, then $R$ too is not GFS).

Given two binary relations $R_s$ and $R_t$ on $Q$,
we define 
 $S(R_s,R_t) = \{ (\langle p, \sigma, r_1 \cdots r_n \rangle , \langle p', \sigma, r'_1 \cdots r'_n \rangle) \, \mid \, p R_s p' \mbox{ and } (r_1 \cdots r_n) \hat{R_t} (r'_1 \cdots r'_n) \},$
where $\hat{R_t}$ is the standard lifting of $R_t \subseteq Q \times Q$ to $\hat{R_t} \subseteq Q^n \times Q^n$.
Informally,
a transition $t'$ is added to the automaton if there exists already a transition $t$ 
s.t.\ its source state is $R_s$-larger than the source state of $t'$,
and its target states are $\hat{R_t}$-larger than the target states of $t'$.
Theorem~\ref{thrm:dw_trace} below proves that $S(\dwtraceinclusionlarg, \dwtraceinclusion)$ is GFS.
Since the GFS property is downward closed,
it follows that 
$S(\dwtraceinclusionlarg, \dwsim)$,
$S(\dwtraceinclusionlarg, \id)$,
$S(\dwsimlarg, \dwtraceinclusion)$,
$S(\dwsimlarg, \dwsimlarg)$,
$S(\dwsimlarg, \id)$,
$S(\id, \dwtraceinclusion)$
and $S(\id, \dwsim)$
too are GFS.
In Theorem~\ref{thrm:up_trace} (see Appendix~\ref{appendix:proofs} for a proof), 
we prove that $S(\uptraceinclusion{id},\uptraceinclusionlarg{id})$ is GFS. 
Thus it follows that
$S(\uptraceinclusion{id},\upsimlarg{id})$,
$S(\uptraceinclusion{id},\id)$,
$S(\upsim{id},\uptraceinclusionlarg{id})$,
$S(\upsim{id},\upsimlarg{id})$,
$S(\upsim{id},\id)$,
$S(\id,\uptraceinclusionlarg{id})$
and $S(\id,\upsimlarg{id})$
too are GFS.

\vspace{0.2cm}
\begin{theorem} \label{thrm:dw_trace}
 $S(\dwtraceinclusionlarg, \dwtraceinclusion)$ is GFS.
\end{theorem}
\begin{proof}
 Let $\A$ be a TDTA and $A_S = Sat(\A,S(\dwtraceinclusionlarg, \dwtraceinclusion))$.
 We will use induction on $n \geq 1$ to show that for every tree $\tree$ of height $n$ and 
 every run $\run_S$ of $\A_S$ s.t.\ $\tree \stackrel{\run_S}{\Longrightarrow} p$, for some state $p$,
 there exists a run $\run$ of $\A$ s.t.\ $\tree \stackrel{\run}{\Longrightarrow} p$.
 This shows, in particular, that $A_S \subseteq A$.
 
 In the base case $n=1$,
 $\tree$ is a leaf-node $\sigma$,
 for some $\sigma \in \Sigma$. 
 Thus for every run $\run_S$ of $\A_S$
 such that $\tree \stackrel{\run_S}{\Longrightarrow} p$,
 for some state $p$,
 there exists $\langle p, \sigma, \spstate \rangle \in \delta_S$.
 By the definition of $\delta_S$,
 there exists $\langle q, \sigma, \spstate \rangle \in \delta$ s.t.\ $q \dwtraceinclusion p$.
 Consequently, there exists a run $\run$ in $\A$ s.t.\ $\tree \stackrel{\run}{\Longrightarrow} q$.
 By $q \dwtraceinclusion p$,
 there also exists a run $\run'$ of $\A$ s.t.\ $\tree \stackrel{\run'}{\Longrightarrow} p$.
 
 For the induction step,
 let $\tree$ be a tree of height $n>1$ and $a$ its root symbol.
 Thus for every run $\run_S$ of $\A_S$ 
 s.t.\ $\tree \stackrel{\run_S}{\Longrightarrow} p$, for some state $p$,
 there exist $\langle p, a, q_1 \ldots q_{\#(a)} \rangle \in \delta_S$
 and, for each $i:(1 \leq i \leq \#(a))$, 
 a run $\run_{S_i}$ of $\A_S$ s.t.\ $\tree_i \stackrel{\run_{S_i}}{\Longrightarrow} q_i$.
 By the definition of $\delta_S$,
 there exists $\langle p',a,q'_1 \ldots q'_{\#(a)} \rangle \in \delta$
 s.t.\ $p' \dwtraceinclusion p$ and, for every $i:(1 \leq i \leq \#(a))$, 
 $q'_i \dwtraceinclusionlarg q_i$.
 Applying the induction hypothesis to each of the subtrees $\tree_i$,
 we know that for every $\tree_i$-run $\run_{S_i}$ of $\A_S$ ending in $q_i$ 
 there is also a $\tree_i$-run $\run_i$ of $\A$ ending in $q_i$.
 And since $q'_i \dwtraceinclusionlarg q_i$ for every $i:(1 \leq i \leq \#(a))$,
 for each $\tree_i$ there exists a run $\run'_i$ of $\A$ 
 s.t.\ $\tree_i \stackrel{\run'_i}{\Longrightarrow} q'_i$.
 Since there exists $\langle p', a, q'_1 \ldots q'_{\#(a)} \rangle \in \delta$,
 we obtain that there is a run $\run''$ of $\A$ s.t.\ $\tree \stackrel{\run''}{\Longrightarrow} p'$.
 From $p' \dwtraceinclusion p$, 
 it follows that there is also a run $\run'''$ of $\A$ 
 s.t.\ $\tree \stackrel{\run'''}{\Longrightarrow} p$. 
%  \qed
\end{proof}

\begin{theorem} \label{thrm:up_trace}
 $S(\uptraceinclusion{id},\uptraceinclusionlarg{id})$ is GFS.
 \label{theorem:up_trace_inc_gfs}
\end{theorem}

The counterexample in Fig.~\ref{fig:GFS_counterex_dwAndup} shows that
$S(\dwequiv, \upequiv\!\!\!(R))$ is not GFS for any relation $R \!\subseteq\! Q \!\times\! Q$.
The remaining counterexamples can be found in Appendix~\ref{appendix:proofs}:
\begin{itemize}
 \item Figure~\ref{fig:GFS_counterex_idAndupofdw}
shows that $S(\id, \upequiv\!\!\!(\dwequiv))$ is not GFS.
 \item Figure~\ref{fig:GFS_counterex_upofdwAndid} shows that $S(\upequiv\!\!(\dwequiv), \id)$ is not GFS.
 \item Figure~\ref{fig:GFS_counterex_upAnddw} is inspired by an example for a 
similar result for linear trees (i.e., words)~\cite{lorenzomayr:reductionAut2006}.
It shows that $S(\upequiv\!\!\!(R), \dwequiv)$ is not GFS for any relation $R \subseteq Q \times Q$.
\end{itemize}

\begin{small}
\begin{figure}[htbp]
 \begin{center}
%  \qquad \qquad \qquad \quad
 \begin{large}
  \scalebox{0.7}{\input{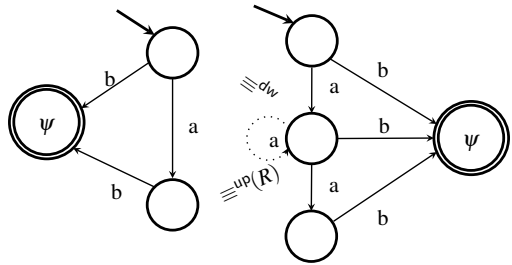}}
 \end{large}
%  \vspace{-10pt}
 \end{center}
 \caption{$S(\dwequiv, \upequiv\!\!\!(R))$ is not GFS for any relation $R \subseteq Q \!\times\! Q$: 
 if we add the dotted transition, 
 the linear tree $aaab$ is now accepted.
 \!The symbol $b$ has rank $0$ and $a$ rank $1$.}
 \label{fig:GFS_counterex_dwAndup}
\end{figure}
\end{small}

In Figure~\ref{fig:GFS_relations} we present a table that summarizes these results.
The negative results follow from the counterexamples given and the fact that the GFS property is downward closed.

\begin{figure}[htbp]
% \vspace{-5pt}
 \begin{center}
% \qquad \qquad \qquad
      \begin{tabular}{c|ccccccc}
	  $S$	   					& $\id$		&	$\dwsim$	& 	$\dwtraceinclusion$	&	$\upsimlarg{\id}$ & 	$\uptraceinclusionlarg{\id}$ &	$\upsimlarg{\dwsimlarg}$ &	$\uptraceinclusionlarg{\dwtraceinclusionlarg}$ 		\\
	  \hline
	  $\id$						& $\tickOK$	& 	$\tickOK$	&	$\tickOK$		&	$\tickOK$	  &	$\tickOK$	  	     &	$\tickNO$		 &	$\tickNO$	\\
% 	  \hline
	  $\dwsimlarg$					& $\tickOK$	& 	$\tickOK$	&	$\tickOK$		&	$\tickNO$	  &	$\tickNO$		     &	$\tickNO$		 &	$\tickNO$	\\
% 	  \hline
	  $\dwtraceinclusionlarg$			& $\tickOK$	& 	$\tickOK$	&	$\tickOK$		&	$\tickNO$	  &	$\tickNO$		     &	$\tickNO$		 &	$\tickNO$	\\
	  $\upsim{\id}$					& $\tickOK$ 	& 	$\tickNO$	&	$\tickNO$		&	$\tickOK$	  &	$\tickOK$		     &	$\tickNO$		 &	$\tickNO$	\\
	  $\uptraceinclusion{\id}$			& $\tickOK$ 	& 	$\tickNO$	&	$\tickNO$		&	$\tickOK$	  &	$\tickOK$		     &	$\tickNO$		 &	$\tickNO$	\\
	  $\upsim{\dwsim}$				& $\tickNO$	& 	$\tickNO$	& 	$\tickNO$		& 	$\tickNO$	  & 	$\tickNO$		     &  $\tickNO$		 & 	$\tickNO$	\\
	  $\uptraceinclusion{\dwtraceinclusion}$	& $\tickNO$	& 	$\tickNO$	& 	$\tickNO$		& 	$\tickNO$	  & 	$\tickNO$		     &  $\tickNO$		 & 	$\tickNO$			  
      \end{tabular}
 \end{center}
  \caption{GFS relations for tree automata.
  Relations which are GFS are marked with $\tickOK$ and those which are not are marked with $\tickNO$.}
  \label{fig:GFS_relations}
\end{figure}

\section{Experimental Results} \label{sec:results}

As we saw in Section~\ref{sec:preliminaires},
the automaton computed by Heavy corresponds to the local minimum of the sequence of reduction techniques used, i.e.,
no smaller automaton can be reached by applying that same sequence of steps again.
The motivation behind saturation is to change this scenario,
since modifying an automaton while preserving its language may leave it
in a state where a different local minimum is reachable by applying Heavy again.
Since saturation adds transitions,
in the end an automaton will either have 
1) the same number of states and the same or larger number of transitions,
2) the same number of states but fewer transitions, or
3) fewer states.
We say that scenarios 2) and 3) correspond to an automaton 'better' than the initial one,
and scenario 1) to a 'worse' one.

Our experiments on test automata consisted of first reducing them with Heavy and then alternating
between saturation and reduction successively until either a fixpoint is reached or the automata becomes 'worse'.
Just like in the case of Heavy, 
there is no ideal order to apply the saturation/reduction techniques,
so we tested multiple possibilities, 
from which we highlight two versions,
Sat1(x,y) and Sat2(x,y),
where $x,y \geq 1$ are the lookaheads used for computing $k$-downward and $k$-upward simulations, 
respectively (see Figure~\ref{fig:sat}).
In both Sat1 and Sat2, 
we chose an order for the operations that ensures that the effect of the saturations is not necessarily 
cancelled by the reductions immediately after.
Intuitively, Sat1 starts by applying both saturations together,
in an attempt to obtain a highly dense automaton where more states may be quotiented.
Sat2, on the other hand, prevents the automaton from becoming too dense, 
by interleaving each downward saturation with the upward reductions it may allow.
Moreover, each upward reduction not only may allow for new downward saturations to be performed, 
but it may also have its effect cancelled if the upward saturation is performed immediately after.
Thus, in Sat2 downward saturation and upward reductions are iterated in an inner loop before performing any upward saturation.
Both versions return the 'best' automaton ever encountered.

\renewcommand{\mytab}{\quad}
\begin{figure}[htbp]
% \vspace{-5pt}
% \begin{scriptsize}
\begin{center}
% \qquad \qquad
\begin{tabular}{l}
 Sat1(x,y)    \\
 Loop:		\\
 \mytab Sat. w/ $S(\transkdwsimlarg{x}, \transkdwsim{x})$ \\
 \mytab Sat. w/ $S(\transkupsim{y}{id},\transkupsimlarg{y}{id})$	\\
 \mytab Quot. w/ $\transkupsim{y}{\id}$	\\
 \mytab Prune w/ $P(\strictupsim{\id},\transkdwsim{x})$ \\
 \mytab Quot. w/ $\transkupsim{y}{\id}$	\\
 \mytab Prune w/ $P(\transkupsim{y}{\dwsim},\strictdwsim)$ \\
 \mytab Run Heavy(x,y)
\end{tabular} \qquad \qquad
\begin{tabular}{l}
 Sat2(x,y)    \\
 Loop:		\\
 \mytab Loop:	\\
 \mytab \mytab Sat. w/ $S(\transkdwsimlarg{x}, \transkdwsim{x})$ \\
 \mytab \mytab Quot. w/ $\transkupsim{y}{\id}$	\\
 \mytab \mytab Prune w/ $P(\transasymkupsim{y}{\id}, \id)$ \\
 \mytab \mytab Prune w/ $P(\strictupsim{\id},\transkdwsim{x})$ \\
 \mytab \mytab Quot. w/ $\transkupsim{y}{\id}$	\\
 \mytab \mytab Prune w/ $P(\transkupsim{y}{\dwsim},\strictdwsim)$ \\
 \mytab Sat. w/ $S(\transkupsim{y}{id},\transkupsimlarg{y}{id})$	\\
 \mytab Run Heavy(x,y)
\end{tabular}
\end{center}
% \end{scriptsize}
\caption{Two saturation-based reduction methods.
% Loops are performed until either a fixpoint is reached or the automaton becomes 'worse'.
Both versions return the 'best' automaton ever encountered. }
\label{fig:sat}
\end{figure}

We tested the different saturation-based reduction methods on a set of 14,498 automata 
(57 states and 266 transitions on avg.)
from the shape analysis tool Forester \cite{tool:forester}.
We can see (Figure~\ref{fig:chart_sat_forester}) 
that, on average, the two versions produced automata containing \emph{both} fewer states and, 
especially, fewer transitions than Heavy alone.
However, this came at the expense of longer running times.
\begin{figure}[htbp]
\centering
\includegraphics[height=5.0cm]{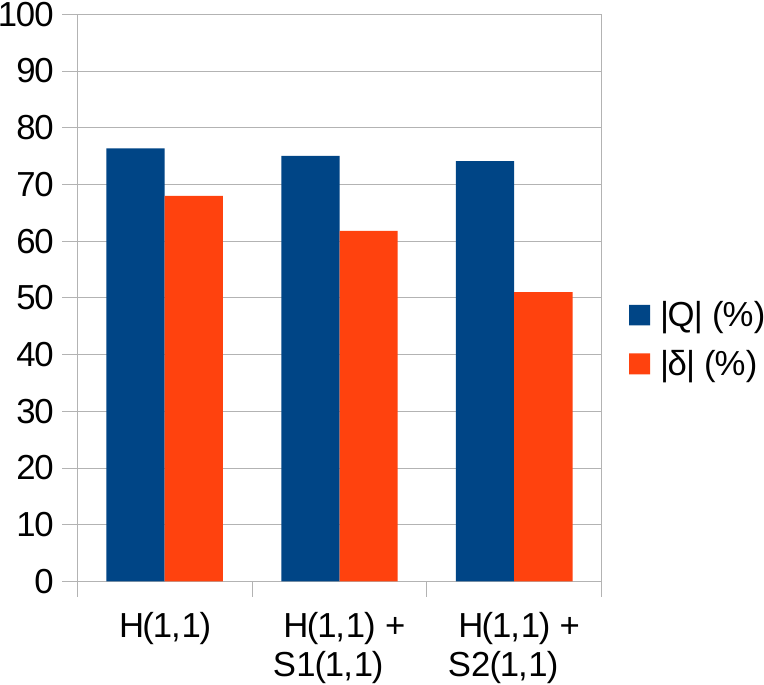} \qquad \quad
\includegraphics[height=5.0cm]{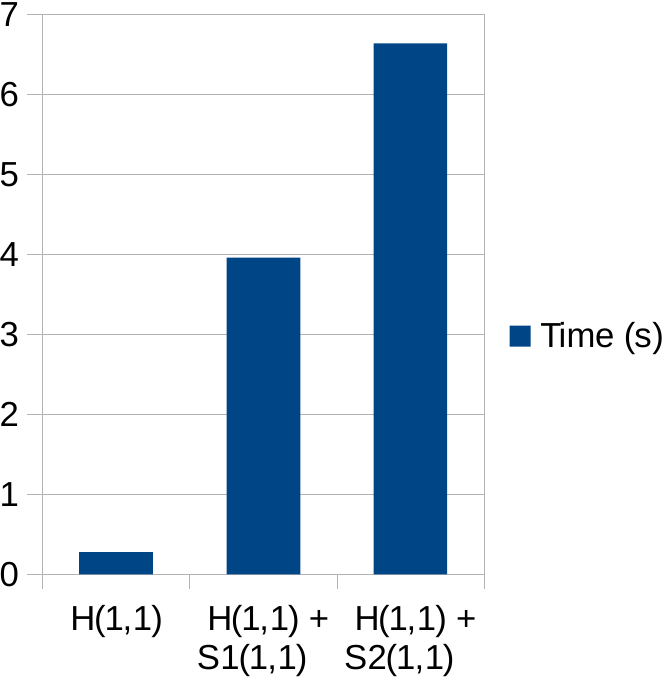}
%  \vspace{-4pt}
 \caption{Reduction of Forester automata using saturation methods. 
 The left chart gives the avg. number of states and transitions that remained (in percentage) 
 after application of each method;
 the right chart compares their running times.
 Heavy(1,1) followed by Sat2(1,1) reduced the automata the most,
 but it was also the slowest method.}
% \end{center}
\label{fig:chart_sat_forester}
\end{figure}

The results that follow focus on the advantage of reducing automata when computing their complement
(for which we use \libvata's implementation of the difference algorithm \cite{hosoya:xml2010}). 
We started by testing on a subset of the Forester sample 
(Fig.~\ref{fig:chart_compl_forester} and Fig.~\ref{fig:chart_compl_forester_app} in App.\ \ref{appendix:experiments}),
and we compared direct complementation with reducing automata (with Heavy(1,1) optionally followed by Sat2(1,1))
prior to the complementation and with a final reduction using Heavy(1,1).
Due to memory reasons, direct complementation was not feasible for large automata.
Thus the sample used is the subset of Forester containing all automata with at most 14 states,
in a total of 760 automata.
As we can see,
all reduction methods yielded significantly smaller complement automata
% both in terms of states and in terms of transitions, 
than direct complementation, on average,
while running either with similar times or substantially faster.
This difference was particularly notorious when the automata were first reduced with both Heavy(1,1) and 
Sat2(1,1),
which, compared to direct complementation,
resulted in automata with fewer states 
(18 vs 27, see Figure~\ref{fig:chart_compl_forester_app} in App.\ \ref{appendix:experiments}) 
and fewer transitions (649 vs 1750) and 
at much lower times (0.02s vs 4.86s).
Applying Heavy(1,1) in the end reduced the automata even more, 
with a very low time cost.
\begin{figure}[htbp]
\begin{center}
 \includegraphics[height=6.0cm]{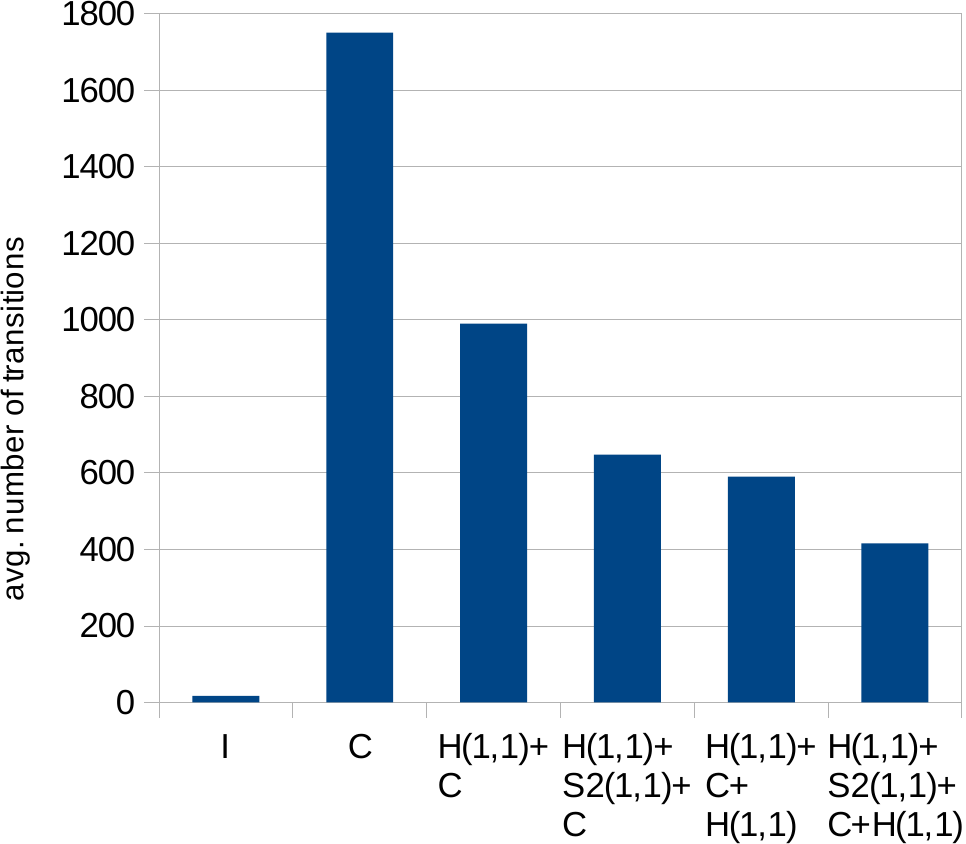} \qquad 
 \includegraphics[height=6.0cm]{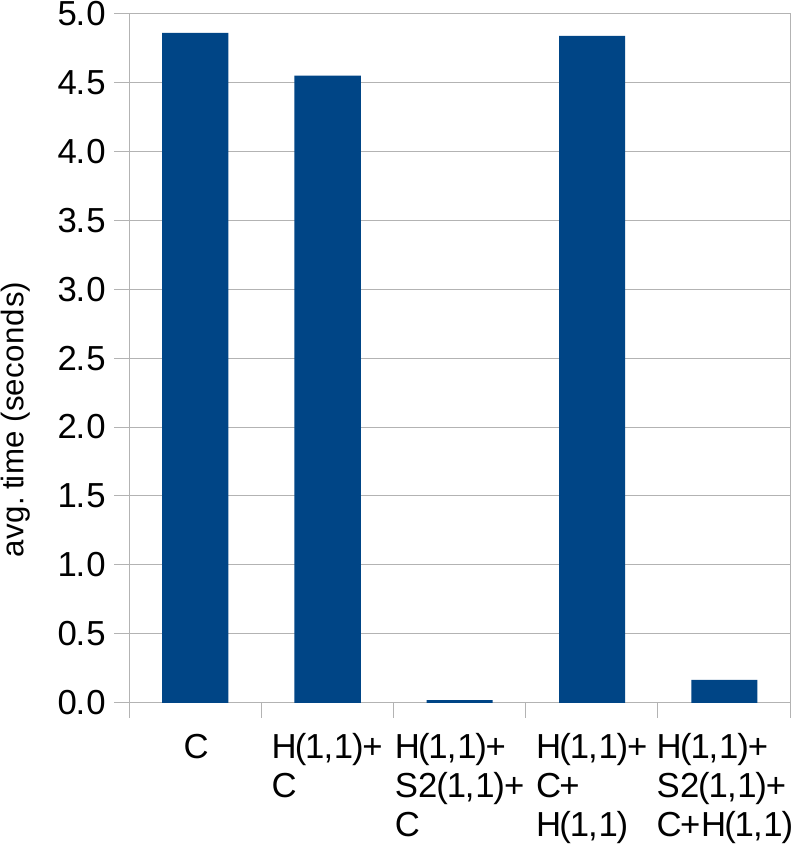} 
\end{center}
 \caption{
 Reducing and complementing Forester automata with at most 14 states.
 The complement automata have fewer transitions and are faster to compute if
 the complementation is preceded by applying Heavy(1,1) and Sat2(1,1) - H(1,1)+S2(1,1)+C -
 or just Heavy(1,1) - H(1,1)+C.
 Applying Heavy(1,1) in the end reduces even more. 
 We include the initial number of transitions (I) for comparison purposes.
 }
\label{fig:chart_compl_forester}
\end{figure}

The next experiments were performed on sets of randomly generated tree automata, 
according to a generalization of the Tabakov-Vardi model of random word automata~\cite{tabakov:model}. 
Given parameters $n, s$, ${\it td}$ (transition density) and ${\it ad}$ (acceptance density), 
it generates tree automata with $n$ states, $s$ symbols (each of rank 2), 
$n*{\it td}$ randomly assigned transitions for each symbol, 
and $n*{\it ad}$ randomly assigned leaf rules. 
Figure~\ref{fig:chart_compl_q=4} 
shows the results of complementing automata with $n=4$ and varying ${\it td}$.
While the automata tested are very small, 
for some values of ${\it td}$ their complements are quite complex 
(more than 400 transitions on average).
As we can see,
applying Heavy not only before but also after the complementation
on average yielded significantly smaller automata,
especially in terms of transitions,
while running with similar times to direct complementation
(all average times were below 0.1s).
Moreover,
the saturation method achieved reductions
% the nondeterminism introduced in the complement automata by the saturation method did pave the way for further reductions 
in the states space which were not possible with Heavy alone.
This came at the cost of higher running times and also of returning automata with more transitions -
but with still far less transitions then those obtained with direct complementation. 
Note that for very dense automata ($td \geq 4.0$), 
the average size of the complement became particularly small.
This is because 
more than half of the automata generated with such $td$ were universal,
and thus their complements were empty.

\begin{figure}[htbp]
\begin{center}
% \quad
% \vspace{0.3cm}
\includegraphics[height=6.5cm]{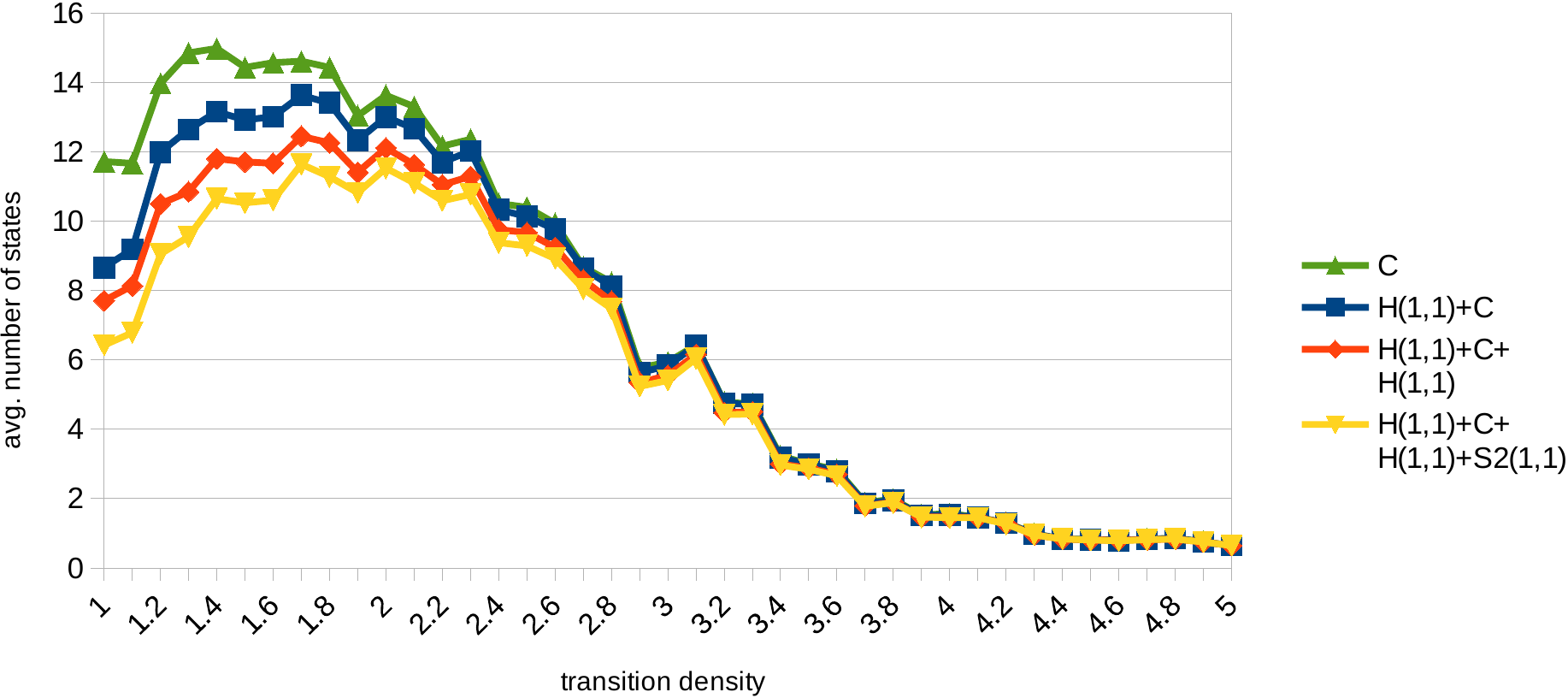} %\qquad
% \vspace{0.3cm}
% \includegraphics[width=12.2cm,height=6.0cm]{chart_compl_q=4_trans-crop.pdf} 
% \vspace{0.3cm}
% \includegraphics[width=12cm,height=6.0cm]{chart_compl_q=4_time-crop.pdf} 
\end{center}
% \vspace{-5pt}
\caption{Reducing and complementing Tabakov-Vardi random tree automata with 4 states.
% We illustrate the average number of states (top chart) and number of transitions (middle chart) 
% after complementation, depending on the type of reduction used.
% The bottom chart presents a comparison in terms of time taken.
Each data point is the average of $300$ automata.
In general, 
applying Heavy(1,1) before the complementation (H(1,1)+C) yielded 
automata with fewer states, on avg., than direct complementation (C).
When Heavy(1,1) is also used after the complementation,
the difference is even more significant -
H(1,1)+C+H(1,1) - and even more when 
Sat2(1,1) is used - H(1,1)+C+H(1,1)+S2(1,1).}
\label{fig:chart_compl_q=4} 
\end{figure}

We also tested our algorithms on random automata with 7 states 
(Figure~\ref{fig:chart_compl_q=7_app} in App.~\ref{appendix:experiments}),
whose complement automata can have, on avg., up to 100 states and more than 30,000 transitions.
As above,
reducing automata with Heavy both before and after the complementation returned
automata with significantly fewer transitions than direct complementation
(3,000 vs 35,000 in some cases),
but the former was clearly slower (avg. times up to 90s) than the latter (avg. times up to 2.5s)
on the automata region where the difference between the two methods was most drastic.
Still, for highly dense automata ($td \geq 4$),
direct complementation was responsible for the highest times recorded
(avg. times between 135s and 2170s).
Due to the size of the complement automata,
the saturation methods revealed to be too slow to be viable in this case.

All experiments were run on an Intel® Core™ i5 @ 3.20GHz x 4 machine with 8GB of RAM using a 64-bit version 
of Ubuntu 16.04.

% \newpage
\bibliographystyle{eptcs}
\bibliography{literature}

\newpage
\appendix
\section{Proofs and Counterexamples} \label{appendix:proofs}

For Lemma~\ref{lemma:uptraceinclusion} and Theorem~\ref{theorem:up_trace_inc_gfs} below
we make use of the following auxiliary definitions.
For every tree $\tree \in \alltrees(\Sigma)$ and every $\tree$-run $\run$,
let $\level_i(\run)$ be the tuple of states that $\run$ visits at depth $i$ in the tree,
read from left to right.
Formally, let $(\node_1,\dots,\node_n)$,
with each $\node_j \in \nat^i$,
be the set of all tree positions of depth $i$ s.t. each $\node_j \in \dom(\run)$,
in lexicographically increasing order.
Then $\level_i(\run) = (\run(\node_1), \dots, \run(\node_n)) \in Q^n$.
We say that $st \in Q^*$ is a subtuple of $\level_i(\run)$,
and write $st \leq \level_i(\run)$, 
if all states in $st$ also appear in $\level_i(\run)$ and in the same order.
By lifting preorders on $Q$ to preorders on $Q^n$,
we can compare tuples of states w.r.t. $\uptraceinclusion{\id}$.

\begin{lemma}
 Let $\A$ be a TDTA and $(p_1,\ldots,p_n)$ and $(q_1,\ldots,q_n)$ two tuples of states of $\A$
 such that $(p_1,\ldots,p_n) \uptraceinclusion{\!\!\,\id} (q_1,\ldots,q_n)$.
 Then, for every $\tree \in \alltrees(\Sigma)$,
 every accepting $\tree$-run $\run$ and every tuple $(\node_1,\dots,\node_n)$ of some leaves of $\run$ of the same depth $i$
 (i.e., $(\node_1,\dots,\node_n) \leq \level_i(\run)$)
 s.t. $(\run(\node_1), \ldots, \run(\node_n)) = (p_1, \ldots, p_n)$,
 there exists an accepting $\tree$-run $\run'$ of $\A$ 
 such that $(\run'(\node_1), \ldots, \run'(\node_n)) = (q_1, \ldots, q_n)$ 
 and $\run'(\node) = \run(\node)$ for every leaf $\node$ of $\run'$ other than $\node_1, \dots, \node_n$.
 \label{lemma:uptraceinclusion}
\end{lemma}
\begin{proof}
 Let $\run$ be an accepting $\tree$-run of $\A$ s.t. $(\run(\node_1),\dots,\run(\node_n)) = (p_1,\dots,p_n)$.
 We say that an accepting $\tree$-run $\run''$ is $i$-good iff 
 i) for every node $\node_j$ of $\run''$, with $j \leq i$, $\run''(\node_j) = q_j$, and
 ii) for every $\node_j$, with $i < j \leq n$, $\run''(\node_j) = p_j$.
 We will show, by induction on $i$, that for every $i$
 there exists an accepting $\tree$-run $\run'''$ which is $i$-good and 
 s.t. $\run'''(\node) = \run(\node)$ for every leaf $\node$ of $\run'''$ other than $\node_1,\dots,\node_n$.
 For the particular case of $i=n$ this proves the lemma.
 
 The base case $i=0$ is trivial,
 since the accepting $\tree$-run $\run$ is $0$-good itself.
 
 For the induction step,
 let $\run_1$ be an accepting $(i-1)$-good $\tree$-run of $\A$.
 If $i > n$, the lemma holds trivially.
 Otherwise, we have $\run_1(\node_i) = p_i \uptraceinclusion{\!\id}\; q_i$
 and thus there exists an accepting $\tree$-run $\run_2$ of $\A$ s.t. $\run_2(\node_i) = q_i$.
 And since the upward trace inclusion is parameterized by $\id$,
 it follows, in particular, that for every leaf $\node$ other than $\node_i$, 
 $\run_2(\node)=\run_1(\node)$.
 Thus, $\run_2$ is an accepting $i$-good $\tree$-run of $\A$.
 Moreover, we have that,  
 on leaves other than $\node_1,\dots,\node_n$,
 the run $\run_2$ coincides with $\run_1$ and consequently, 
 by the induction hypothesis, with $\run$.
\end{proof}

{\bf\noindent Theorem~\ref{thrm:up_trace}}
 $S(\uptraceinclusion{id},\uptraceinclusionlarg{id})$ is GFS.
\begin{proof}
 Let $\A$ be a TDTA and $A_S = Sat(A,S(\uptraceinclusion{id},\uptraceinclusionlarg{id}))$.
 If $\hat{\tree} \in A_S$,
 then there exists an accepting $\hat{\tree}$-run $\hat{\run}$ of $A_S$.
 We will show that there exists an accepting $\hat{\tree}$-run of $\A$, 
 which proves $A_S \subseteq A$.
 
 Let us first define an auxiliary notion.
 For every $\tree \in \alltrees(\Sigma)$ and every $\tree$-run $\run$,
 we say that $\run$ is $i$-good iff it does not contain any transition of $\delta_S - \delta$
 from any position $\node \in \nat^*$ s.t. $|\node| < i$, i.e.,
 all transitions used in the first $i$ levels of the tree are of $A$.
 
 Next, we will show, by induction on $i$, 
 that for every $i$ there exists an accepting $i$-good $\hat{\tree}$-run $\hat{\run}'$ of $A_S$
 s.t. $\level_i(\hat{\run}') = \level_i(\hat{\run})$.
 For $i$ equal to the height of $\hat{\tree}$,
 this implies that there exists an accepting $\hat{\tree}$-run of $\A$.
 
 The base case $i=0$ is trivial, since $\hat{\run}$ is $0$-good itself.
 
 For the induction step,
 let us first define some auxiliary notions.
 For every $\tree \in \alltrees(\Sigma)$ and every $\tree$-run $\run$,
 we say that $\level_{i'}(\run)$ is $j$-good iff $\run$ does not contain a transition of 
 $\delta_S - \delta$ from a state $\run(\node_k)$, 
 s.t. $k \leq j$ and $\run(\node_k)$ is the $k$-th state of $\level_{i'}(\run)$.
 We now say that an accepting $\hat{\tree}$-run $\hat{\run}''$ of $A_S$ is $(i-1,j)$-good iff
 i) it is $(i-1)$-good, 
 ii) $\level_{i-1}(\hat{\run}'')$ is $j$-good, and
 iii) $\level_i(\hat{\run}'') = \level_i(\hat{\run})$.
 
 We will now show, by induction on $j$,
 that for every $j$ there exists an accepting $(i-1,j)$-good $\hat{\tree}$-run of $A_S$.
 Since trees are finitely-branching,
 we have that for a sufficiently large $j$ there is an 
 accepting $\hat{\tree}$-run $\hat{\run}'''$ of $A_S$ which is $i$-good.
 And since, in particular, $\level_i(\hat{\run}''') = \level_i(\hat{\run})$,
 this will conclude the outer induction.
 
 For the base case $(i-1,0)$,
 we know by the hypothesis of the outer induction that there exists an accepting $(i-1)$-good
 $\hat{\tree}$-run $\run_1$ s.t. $\level_{i-1}(\run_1) = \level_{i-1}(\hat{\run})$.
 Then the $\hat{\tree}$-run $\run_2$ which,
 on the levels below $i$, coincides with $\run_1$ and,
 on the levels from $i$ up,
 coincides with $\hat{\run}$ too is accepting and $(i-1)$-good.
 Thus $\run_2$ is $(i-1,0)$-good.
 
 For the induction step,
 let $\run_1$ be an accepting $(i-1,j-1)$-good $\hat{\tree}$-run of $A_S$,
 and let $\run_1'$ be the prefix of $\run_1$ which only uses transitions of $A$.
 $\run_1'$ is thus an accepting run of $A$ over some prefix tree $\hat{\tree}'$ of $\hat{\tree}$.
 Let $\node_j$ be the node of $\hat{\tree}$ s.t. $\run_1'(\node_j)$ is the $j$-th 
 state of $\level_{i-1}(\run_1')$ and $\sigma = \hat{\tree}(\node_j)$ a symbol of rank $r$.
 
 If $r=0$, then $\node_j$ is a leaf of $\hat{\tree}$ and so there exists a transition 
 $\langle \run_1'(\node_j), \sigma, \spstate \rangle$ in $A_S$.
 By the definition of $\delta_S$,
 there exists a transition $\langle p, \sigma, \spstate \rangle$ in $A$ 
 s.t. $\run_1'(\node_j) \uptraceinclusion{\id} \;p$.
 Thus there exists an accepting $\hat{\tree}'$-run $\run_2$ of $\A$ s.t. $\run_2(\node_j) = p$ and
 for any leaf $\node$ of $\run_2$ other than $\node_j$, $\run_2(\node) = \run_1'(\node)$.
 We now obtain a run over $\hat{\tree}$ again by extending $\run_2$ downwards according to $\run_1$, i.e.,
 $\run_2(\node\node') := \run_1(\node\node')$,
 for every leaf $\node$ of $\run_2$ other than $\node_j$ and for every $\node' \in \nat^*$.
 It follows that $\level_i(\run_2) = \level_i(\run_1) = \level_i(\hat{\run})$.
 $\run_2$ is clearly a $(i-1)$-good $\hat{\tree}$-run of $A_S$ and $\level_{i-1}(\run_2)$
 is $j$-good.
 Thus $\run_2$ is an accepting $(i-1,j)$-good $\hat{\tree}$-run of $A_S$.
 
 If $r > 0$,
 then $\node_j$ is not a leaf and so there exists a transition 
 $\langle \run_1'(\node_j), \sigma, \run_1(\node_j1) \dots \run_1(\node_jr) \rangle$ in $A_S$.
 By the definition of $\delta_S$, 
 there exists a transition $trans$: $\langle p, \sigma, q_1 \dots q_r \rangle$ in $A$ s.t.
 $\run_1'(\node_j) \uptraceinclusion{\id} p$ and
 \\ 1) $(q_1 \dots q_r) \uptraceinclusion{\id} (\run_1(\node_j1) \dots \run_1(\node_jr))$.
 From $\run_1'(\node_j) \uptraceinclusion{\id} \; p$ we have that there exists an accepting
 $\hat{\tree}'$-run $\run_2$ of $\A$ s.t. $\run_2(\node_j) = p$ and $\run_2(\node) = \run_1'(\node)$,
 for every leaf $\node$ of $\run_2$ other than $\node_j$.
 Extending $\run_2$ with $trans$ we obtain an accepting run of $A$ s.t. 
 $\run_2(\node_jk) := q_k$
 for each child $\node_jk$ of $\node_j$.
 Applying Lemma~\ref{lemma:uptraceinclusion} to 1),
 we obtain that there exists an accepting run $\run_3$ of $\A$ over the same prefix tree of $\hat{t}$
 as $\run_2$ s.t. 
 2) $\run_3(\node_jk) = \run_1(\node_jk)$ for each child $\node_jk$ of $\node_j$,
 and $\run_3(\node) = \run_2(\node) = \run_1(\node)$ for every leaf $\node$ of $\run_3$ other than
 $\node_j1, \dots, \node_jr$.
 We now obtain a run over $\hat{\tree}$ again by extending $\run_3$ downwards according to $\run_1$, i.e.,
 3) $\run_3(\node\node') := \run_1(\node\node')$,
 for every leaf $\node$ of $\run_3$ other than $\node_j1, \cdots, \node_jr$ and for every $\node' \in \nat^*$.
 $\run_3$ is clearly a $(i-1)$-good $\hat{\tree}$-run of $A_S$ and $\level_{i-1}(\run_3)$ is $j$-good.
 From 2) and 3),
 we obtain that $\level_i(\run_3) = \level_i(\run_1) = \level_i(\hat{\run})$.
 Thus $\run_3$ is an accepting $(i-1,j)$-good $\hat{\tree}$-run of $A_S$.
\end{proof}

\begin{figure}[p]
 \begin{center}
 \begin{large}
  \scalebox{0.7}{\input{GFS-counterex-idAndupofdw.tex}}
 \end{large}
 \end{center}
 \vspace{-5pt}
 \caption{$S(\id, \upequiv\!\!\!(\dwequiv))$ is not GFS: if we add the dotted transitions, 
 the tree $a(b(c),b(c))$ is now accepted.
 The symbols $c$, $b$ and $a$ have ranks $0$, $1$ and $2$, resp.}
 \label{fig:GFS_counterex_idAndupofdw}
\end{figure}

\begin{figure}[p]
 \begin{center}
 \begin{large}
  \scalebox{0.7}{\input{GFS-counterex-upofdwAndid.tex}}
 \end{large}
 \end{center}
 \vspace{-5pt}
 \caption{$S(\upequiv\!\!(\dwequiv), \id)$ is not GFS: if we add the dotted transitions, 
 the tree $a(b,b)$ is now accepted.
 The symbols $b$ and $a$ have ranks $0$ and $1$, respectively.}
 \label{fig:GFS_counterex_upofdwAndid}
\end{figure}

\begin{figure}
 \begin{center}
 \begin{large}
  \scalebox{0.7}{% Graphic for TeX using PGF
% Title: /home/ricardo/Documents/PhD/Tree_Automata/svn/tree_lookahead/Saturation/GFS_counterex_up&dw.dia
% Creator: Dia v0.97.3
% CreationDate: Wed Jul 27 14:55:49 2016
% For: ricardo
% \usepackage{tikz}
% The following commands are not supported in PSTricks at present
% We define them conditionally, so when they are implemented,
% this pgf file will use them.
\ifx\du\undefined
  \newlength{\du}
\fi
\setlength{\du}{15\unitlength}
\begin{tikzpicture}
\pgftransformxscale{1.000000}
\pgftransformyscale{-1.000000}
\definecolor{dialinecolor}{rgb}{0.000000, 0.000000, 0.000000}
\pgfsetstrokecolor{dialinecolor}
\definecolor{dialinecolor}{rgb}{1.000000, 1.000000, 1.000000}
\pgfsetfillcolor{dialinecolor}
\definecolor{dialinecolor}{rgb}{1.000000, 1.000000, 1.000000}
\pgfsetfillcolor{dialinecolor}
\pgfpathellipse{\pgfpoint{26.884426\du}{4.372139\du}}{\pgfpoint{0.913526\du}{0\du}}{\pgfpoint{0\du}{0.910869\du}}
\pgfusepath{fill}
\pgfsetlinewidth{0.100000\du}
\pgfsetdash{}{0pt}
\pgfsetdash{}{0pt}
\pgfsetmiterjoin
\definecolor{dialinecolor}{rgb}{0.000000, 0.000000, 0.000000}
\pgfsetstrokecolor{dialinecolor}
\pgfpathellipse{\pgfpoint{26.884426\du}{4.372139\du}}{\pgfpoint{0.913526\du}{0\du}}{\pgfpoint{0\du}{0.910869\du}}
\pgfusepath{stroke}
% setfont left to latex
\definecolor{dialinecolor}{rgb}{0.000000, 0.000000, 0.000000}
\pgfsetstrokecolor{dialinecolor}
\node at (26.884426\du,4.567139\du){};
% setfont left to latex
\definecolor{dialinecolor}{rgb}{0.000000, 0.000000, 0.000000}
\pgfsetstrokecolor{dialinecolor}
\node[anchor=west] at (25.132300\du,1.844450\du){$\upequiv\!\!\!(R)$};
\definecolor{dialinecolor}{rgb}{1.000000, 1.000000, 1.000000}
\pgfsetfillcolor{dialinecolor}
\pgfpathellipse{\pgfpoint{22.686826\du}{1.670893\du}}{\pgfpoint{0.913526\du}{0\du}}{\pgfpoint{0\du}{0.910869\du}}
\pgfusepath{fill}
\pgfsetlinewidth{0.100000\du}
\pgfsetdash{}{0pt}
\pgfsetdash{}{0pt}
\pgfsetmiterjoin
\definecolor{dialinecolor}{rgb}{0.000000, 0.000000, 0.000000}
\pgfsetstrokecolor{dialinecolor}
\pgfpathellipse{\pgfpoint{22.686826\du}{1.670893\du}}{\pgfpoint{0.913526\du}{0\du}}{\pgfpoint{0\du}{0.910869\du}}
\pgfusepath{stroke}
% setfont left to latex
\definecolor{dialinecolor}{rgb}{0.000000, 0.000000, 0.000000}
\pgfsetstrokecolor{dialinecolor}
\node at (22.686826\du,1.865893\du){};
\pgfsetlinewidth{0.100000\du}
\pgfsetdash{}{0pt}
\pgfsetdash{}{0pt}
\pgfsetbuttcap
{
\definecolor{dialinecolor}{rgb}{0.000000, 0.000000, 0.000000}
\pgfsetfillcolor{dialinecolor}
% was here!!!
\pgfsetarrowsend{stealth}
\definecolor{dialinecolor}{rgb}{0.000000, 0.000000, 0.000000}
\pgfsetstrokecolor{dialinecolor}
\draw (20.580000\du,0.269064\du)--(21.885193\du,1.137507\du);
}
\pgfsetlinewidth{0.050000\du}
\pgfsetdash{}{0pt}
\pgfsetdash{}{0pt}
\pgfsetbuttcap
{
\definecolor{dialinecolor}{rgb}{0.000000, 0.000000, 0.000000}
\pgfsetfillcolor{dialinecolor}
% was here!!!
\pgfsetarrowsend{stealth}
\definecolor{dialinecolor}{rgb}{0.000000, 0.000000, 0.000000}
\pgfsetstrokecolor{dialinecolor}
\draw (22.686800\du,2.581760\du)--(22.692200\du,5.029260\du);
}
\pgfsetlinewidth{0.050000\du}
\pgfsetdash{}{0pt}
\pgfsetdash{}{0pt}
\pgfsetbuttcap
{
\definecolor{dialinecolor}{rgb}{0.000000, 0.000000, 0.000000}
\pgfsetfillcolor{dialinecolor}
% was here!!!
\pgfsetarrowsend{stealth}
\definecolor{dialinecolor}{rgb}{0.000000, 0.000000, 0.000000}
\pgfsetstrokecolor{dialinecolor}
\draw (26.238500\du,5.016220\du)--(24.552600\du,6.249850\du);
}
\pgfsetlinewidth{0.050000\du}
\pgfsetdash{}{0pt}
\pgfsetdash{}{0pt}
\pgfsetbuttcap
{
\definecolor{dialinecolor}{rgb}{0.000000, 0.000000, 0.000000}
\pgfsetfillcolor{dialinecolor}
% was here!!!
\pgfsetarrowsend{stealth}
\definecolor{dialinecolor}{rgb}{0.000000, 0.000000, 0.000000}
\pgfsetstrokecolor{dialinecolor}
\draw (23.332700\du,2.314980\du)--(26.051557\du,3.889741\du);
}
% setfont left to latex
\definecolor{dialinecolor}{rgb}{0.000000, 0.000000, 0.000000}
\pgfsetstrokecolor{dialinecolor}
\node[anchor=west] at (21.907796\du,3.680610\du){c};
\definecolor{dialinecolor}{rgb}{1.000000, 1.000000, 1.000000}
\pgfsetfillcolor{dialinecolor}
\pgfpathellipse{\pgfpoint{23.198305\du}{6.241876\du}}{\pgfpoint{1.178805\du}{0\du}}{\pgfpoint{0\du}{1.175376\du}}
\pgfusepath{fill}
\pgfsetlinewidth{0.100000\du}
\pgfsetdash{}{0pt}
\pgfsetdash{}{0pt}
\pgfsetmiterjoin
\definecolor{dialinecolor}{rgb}{0.000000, 0.000000, 0.000000}
\pgfsetstrokecolor{dialinecolor}
\pgfpathellipse{\pgfpoint{23.198305\du}{6.241876\du}}{\pgfpoint{1.178805\du}{0\du}}{\pgfpoint{0\du}{1.175376\du}}
\pgfusepath{stroke}
% setfont left to latex
\definecolor{dialinecolor}{rgb}{0.000000, 0.000000, 0.000000}
\pgfsetstrokecolor{dialinecolor}
\node at (23.198305\du,6.436876\du){$\spstate$};
\pgfsetlinewidth{0.100000\du}
\pgfsetdash{}{0pt}
\pgfsetdash{}{0pt}
\pgfsetmiterjoin
\definecolor{dialinecolor}{rgb}{0.000000, 0.000000, 0.000000}
\pgfsetstrokecolor{dialinecolor}
\pgfpathellipse{\pgfpoint{23.207067\du}{6.249851\du}}{\pgfpoint{1.345467\du}{0\du}}{\pgfpoint{0\du}{1.321151\du}}
\pgfusepath{stroke}
% setfont left to latex
\definecolor{dialinecolor}{rgb}{0.000000, 0.000000, 0.000000}
\pgfsetstrokecolor{dialinecolor}
\node at (23.207067\du,6.444851\du){};
\pgfsetlinewidth{0.050000\du}
\pgfsetdash{{\pgflinewidth}{0.200000\du}}{0cm}
\pgfsetdash{{\pgflinewidth}{0.200000\du}}{0cm}
\pgfsetbuttcap
{
\definecolor{dialinecolor}{rgb}{0.000000, 0.000000, 0.000000}
\pgfsetfillcolor{dialinecolor}
% was here!!!
\pgfsetarrowsend{stealth}
\definecolor{dialinecolor}{rgb}{0.000000, 0.000000, 0.000000}
\pgfsetstrokecolor{dialinecolor}
\draw (30.692486\du,1.101903\du)--(30.960052\du,1.745985\du);
}
% setfont left to latex
\definecolor{dialinecolor}{rgb}{0.000000, 0.000000, 0.000000}
\pgfsetstrokecolor{dialinecolor}
\node[anchor=west] at (30.686413\du,0.857468\du){a};
% setfont left to latex
\definecolor{dialinecolor}{rgb}{0.000000, 0.000000, 0.000000}
\pgfsetstrokecolor{dialinecolor}
\node[anchor=west] at (25.636217\du,5.988070\du){c};
% setfont left to latex
\definecolor{dialinecolor}{rgb}{0.000000, 0.000000, 0.000000}
\pgfsetstrokecolor{dialinecolor}
\node[anchor=west] at (24.681500\du,2.777220\du){a};
\pgfsetlinewidth{0.100000\du}
\pgfsetdash{}{0pt}
\pgfsetdash{}{0pt}
\pgfsetbuttcap
{
\definecolor{dialinecolor}{rgb}{0.000000, 0.000000, 0.000000}
\pgfsetfillcolor{dialinecolor}
% was here!!!
\pgfsetarrowsend{stealth}
\definecolor{dialinecolor}{rgb}{0.000000, 0.000000, 0.000000}
\pgfsetstrokecolor{dialinecolor}
\draw (28.039200\du,0.335822\du)--(29.344400\du,1.204260\du);
}
\definecolor{dialinecolor}{rgb}{1.000000, 1.000000, 1.000000}
\pgfsetfillcolor{dialinecolor}
\pgfpathellipse{\pgfpoint{30.046526\du}{1.745985\du}}{\pgfpoint{0.913526\du}{0\du}}{\pgfpoint{0\du}{0.910869\du}}
\pgfusepath{fill}
\pgfsetlinewidth{0.100000\du}
\pgfsetdash{}{0pt}
\pgfsetdash{}{0pt}
\pgfsetmiterjoin
\definecolor{dialinecolor}{rgb}{0.000000, 0.000000, 0.000000}
\pgfsetstrokecolor{dialinecolor}
\pgfpathellipse{\pgfpoint{30.046526\du}{1.745985\du}}{\pgfpoint{0.913526\du}{0\du}}{\pgfpoint{0\du}{0.910869\du}}
\pgfusepath{stroke}
% setfont left to latex
\definecolor{dialinecolor}{rgb}{0.000000, 0.000000, 0.000000}
\pgfsetstrokecolor{dialinecolor}
\node at (30.046526\du,1.940985\du){};
\pgfsetlinewidth{0.050000\du}
\pgfsetdash{}{0pt}
\pgfsetdash{}{0pt}
\pgfsetbuttcap
{
\definecolor{dialinecolor}{rgb}{0.000000, 0.000000, 0.000000}
\pgfsetfillcolor{dialinecolor}
% was here!!!
\pgfsetarrowsend{stealth}
\definecolor{dialinecolor}{rgb}{0.000000, 0.000000, 0.000000}
\pgfsetstrokecolor{dialinecolor}
\draw (30.046600\du,2.656850\du)--(30.031000\du,4.940700\du);
}
% setfont left to latex
\definecolor{dialinecolor}{rgb}{0.000000, 0.000000, 0.000000}
\pgfsetstrokecolor{dialinecolor}
\node[anchor=west] at (30.202000\du,3.900730\du){c};
\definecolor{dialinecolor}{rgb}{1.000000, 1.000000, 1.000000}
\pgfsetfillcolor{dialinecolor}
\pgfpathellipse{\pgfpoint{30.022305\du}{6.253876\du}}{\pgfpoint{1.178805\du}{0\du}}{\pgfpoint{0\du}{1.175376\du}}
\pgfusepath{fill}
\pgfsetlinewidth{0.100000\du}
\pgfsetdash{}{0pt}
\pgfsetdash{}{0pt}
\pgfsetmiterjoin
\definecolor{dialinecolor}{rgb}{0.000000, 0.000000, 0.000000}
\pgfsetstrokecolor{dialinecolor}
\pgfpathellipse{\pgfpoint{30.022305\du}{6.253876\du}}{\pgfpoint{1.178805\du}{0\du}}{\pgfpoint{0\du}{1.175376\du}}
\pgfusepath{stroke}
% setfont left to latex
\definecolor{dialinecolor}{rgb}{0.000000, 0.000000, 0.000000}
\pgfsetstrokecolor{dialinecolor}
\node at (30.022305\du,6.448876\du){$\spstate$};
\pgfsetlinewidth{0.100000\du}
\pgfsetdash{}{0pt}
\pgfsetdash{}{0pt}
\pgfsetmiterjoin
\definecolor{dialinecolor}{rgb}{0.000000, 0.000000, 0.000000}
\pgfsetstrokecolor{dialinecolor}
\pgfpathellipse{\pgfpoint{30.031067\du}{6.261851\du}}{\pgfpoint{1.345467\du}{0\du}}{\pgfpoint{0\du}{1.321151\du}}
\pgfusepath{stroke}
% setfont left to latex
\definecolor{dialinecolor}{rgb}{0.000000, 0.000000, 0.000000}
\pgfsetstrokecolor{dialinecolor}
\node at (30.031067\du,6.456851\du){};
\pgfsetlinewidth{0.050000\du}
\pgfsetdash{{\pgflinewidth}{0.200000\du}}{0cm}
\pgfsetdash{{\pgflinewidth}{0.200000\du}}{0cm}
\pgfsetbuttcap
{
\definecolor{dialinecolor}{rgb}{0.000000, 0.000000, 0.000000}
\pgfsetfillcolor{dialinecolor}
% was here!!!
\pgfsetarrowsstart{stealth}
\definecolor{dialinecolor}{rgb}{0.000000, 0.000000, 0.000000}
\pgfsetstrokecolor{dialinecolor}
\pgfpathmoveto{\pgfpoint{30.890465\du}{1.397399\du}}
\pgfpatharc{104}{-193}{0.672245\du and 0.672245\du}
\pgfusepath{stroke}
}
% setfont left to latex
\definecolor{dialinecolor}{rgb}{0.000000, 0.000000, 0.000000}
\pgfsetstrokecolor{dialinecolor}
\node[anchor=west] at (27.576917\du,3.138946\du){\rotatebox{45}{$\dwequiv$}};
\end{tikzpicture}}
 \end{large}
 \end{center}
 \vspace{-5pt}
 \caption{$S(\upequiv\!\!\!(R), \dwequiv)$ is not GFS for any relation $R \subseteq Q \times Q$
 (example adapted from~\cite{lorenzomayr:reductionAut2006}): 
 if we add the dotted transition, 
 the linear tree $aac$ is now accepted.
 The symbol $c$ has rank $0$ and $a$ has rank $1$.}
 \label{fig:GFS_counterex_upAnddw}
\end{figure}

\newpage
\section{Variants of the Optimization to the Lookahead Simulation Game} \label{appendix:caching}

In Section~\ref{sec:efficient} we presented an optimization to the computation of the $k$-lookahead downward simulation
based on the caching of attacks in the simulation game between Spoiler and Duplicator.
In this appendix we present two alternative versions to this optimization,
where we change the scope of the attacks cached.

\paragraph{Local caching of Spoiler's attacks.}
Whenever Spoiler uses a transition $t$ in an attack,
Duplicator can memorize which states in the automaton are able to defend 
against the target states of $t$.
In Fig.~\ref{fig:opt_automata} from Section~\ref{sec:efficient}, 
in a round of the simulation game from 
$(p_2,q_2)$, Spoiler is attempting the attack $d(e,e)$ leading to $p_5,p_7$.
Duplicator tries responding with a $d$-transition to $(q_4,q_5)$,
and since there is a $e$-transition from $q_4$ to $q_7$ and $p_5 \,W\, q_7$,
Duplicator caches the information that, 
against $q_4$, the first sub-attack is a \emph{bad} one.
However, $q_5$ can only read $f$ and so Duplicator will have to try a different defence.
Duplicator now tries the $d$-transition leading to $q_4$ and $q_6$ instead.
Thanks to the information recorded, 
Duplicator now only needs to find a defence from $q_6$ against $p_6$,
which exists since $q_6$ goes to $q_8$ by $e$ and $p_5 \,W\, q_8$,
and so Duplicator declares victory against this particular attack.

Conversely, if the game configuration was $(p_2,q_3)$,
after trying to defend against the attack $c(e,e)$ using the $c$-transition to $q_9$ and $q_{10}$, 
Duplicator could reuse the information that the sub-attack $e$ is \emph{good} against $q_9$ 
when trying the $c$-transition to $q_9$ and $q_{11}$.

\paragraph{Global caching of Spoiler's attacks.}
Here we expand the scope to the entire $W$-refinement.
E.g., the \emph{good} attacks from $p_2$ against $q_2$ or against $q_3$ can be 
recalled even when a game from a different configuration, say, $(p_8,q_1)$ is played.
However, the information about the \emph{bad} attack from, say, 
$p_3$ against $q_4$ cannot be used outside of the local game in which it was saved,
since Duplicator could only defend against it based on the state of $W$ at the time.
Note the asymmetry between \emph{good} and \emph{bad} attacks:
\emph{good} attacks remain \emph{good} for the rest of the entire computation, but
\emph{bad} attacks may become \emph{good} after $W$ changes.

\newpage
\section{More Charts from the Experimental Results} \label{appendix:experiments}

\begin{figure}[htbp]
\begin{center}
 \includegraphics[height=6.0cm]{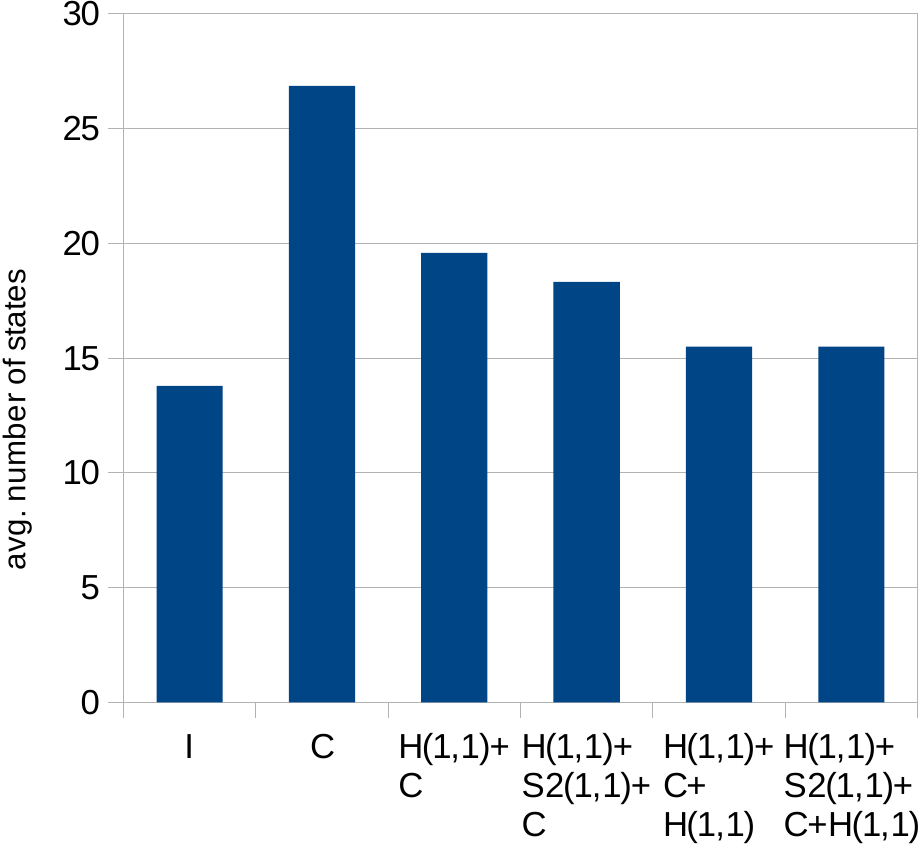} 
\end{center}
 \caption{
 Reducing and complementing Forester automata with at most 14 states.
 The complement automata have fewer states if
 the complementation is preceded by applying Heavy(1,1) and Sat2(1,1) - H(1,1)+S2(1,1)+C -
 or just Heavy(1,1) - H(1,1)+C.
 Applying Heavy(1,1) in the end reduces even more. 
 We include the initial number of states (I) for comparison purposes.
 }
\label{fig:chart_compl_forester_app}
\end{figure}

\begin{figure}
\begin{center}
% \vspace{0.3cm}
% \includegraphics[width=12cm,height=6.0cm]{chart_compl_q=4_states-crop.pdf} %\qquad
% \vspace{0.3cm}
\includegraphics[height=6.5cm]{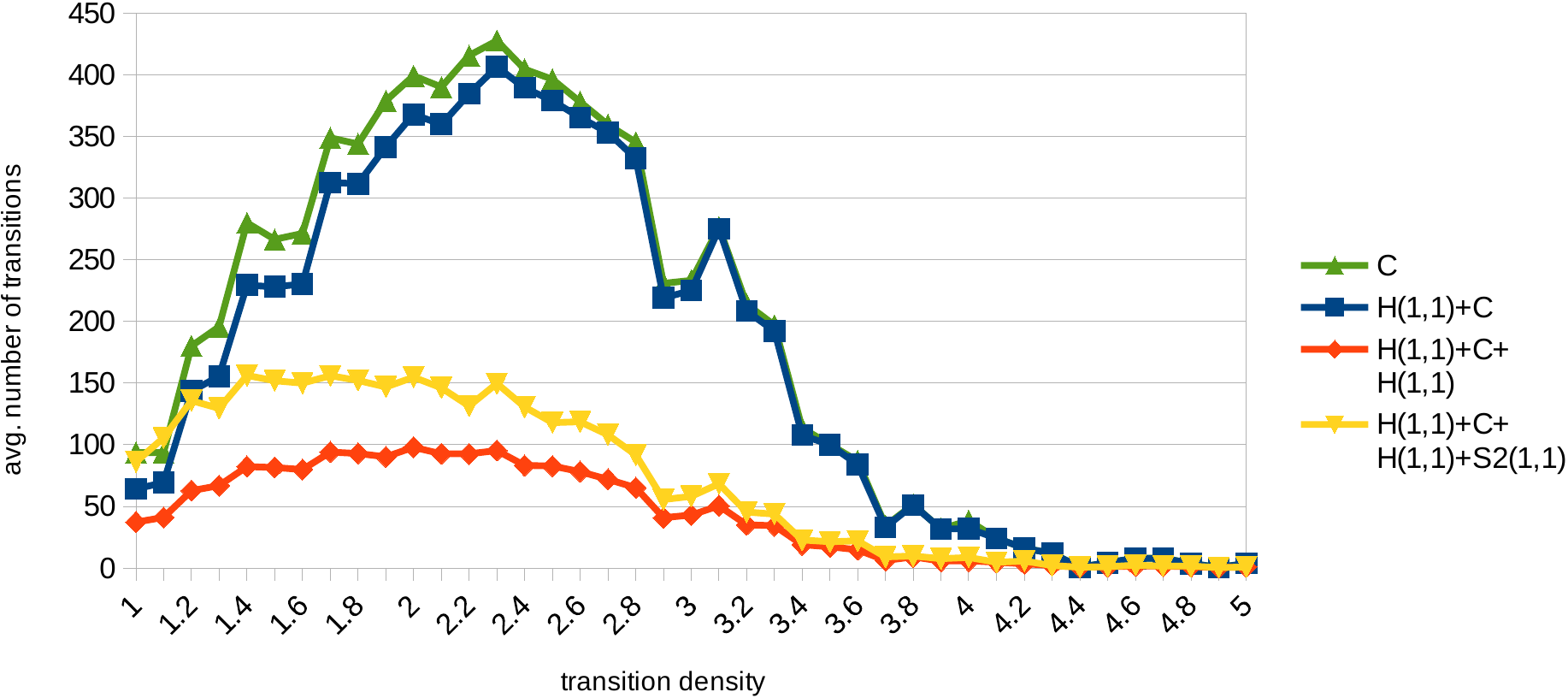} 
% \\
\includegraphics[height=6.5cm]{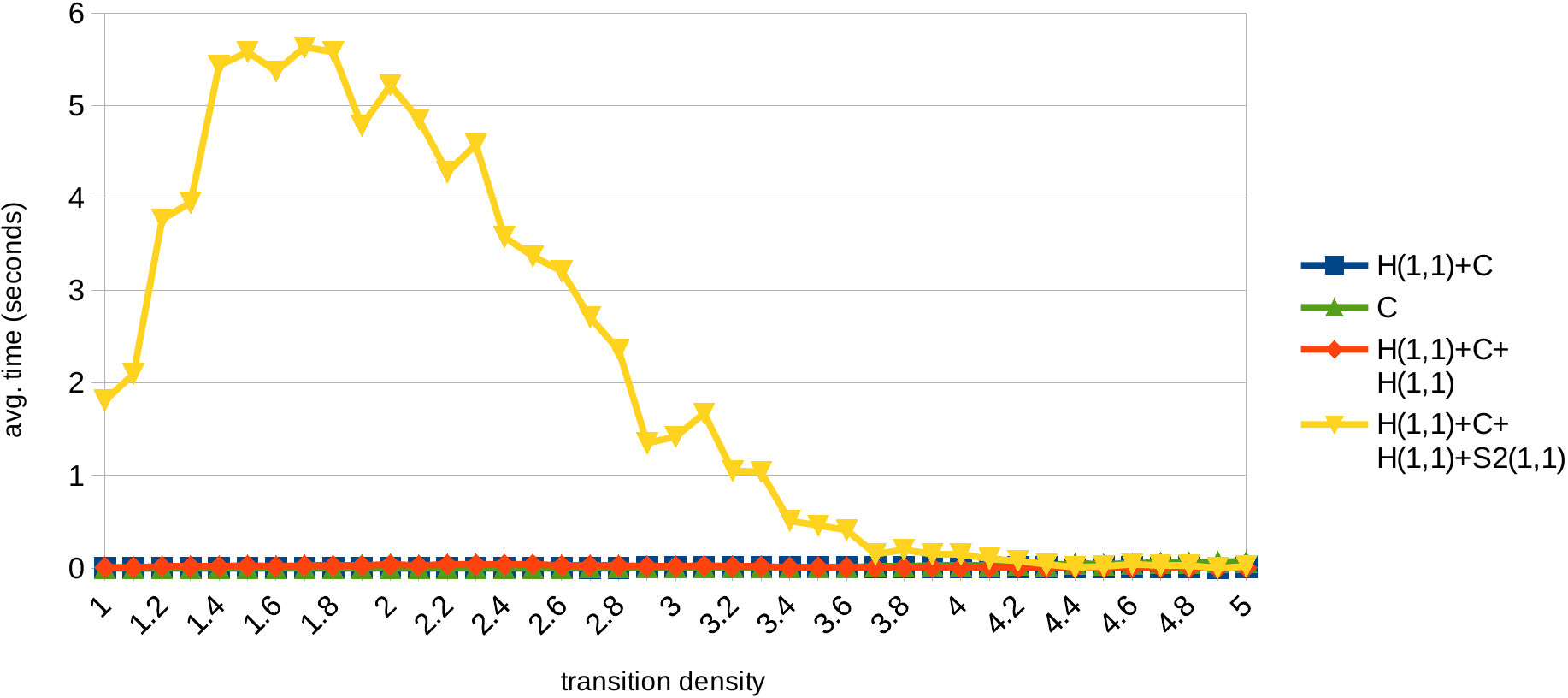} 
\end{center}
\caption{Reducing and complementing Tabakov-Vardi random tree automata with 4 states.
% We illustrate the average number of states (top chart) and number of transitions (middle chart) 
% after complementation, depending on the type of reduction used.
% The bottom chart presents a comparison in terms of time taken.
Each data point is the average of $300$ automata.
In general, 
applying Heavy(1,1) before the complementation (H(1,1)+C) yielded 
automata with fewer states and transitions, on average, than direct complementation (C).
When Heavy is used both before and after the complementation,
the difference is even more significant:
H(1,1)+C+H(1,1) produced automata with less than $1/3$ of the transitions of C for nearly all values of $td$.
Running Heavy followed by Sat2 after the complementation (H(1,1)+C+H(1,1)+S2(1,1))
offered a trade-off between reduction in the states space and in the number of transitions (as well as in time).}
\label{fig:chart_compl_q=4_app} 
\end{figure}

\begin{figure}
\begin{center}
\includegraphics[height=6.5cm]{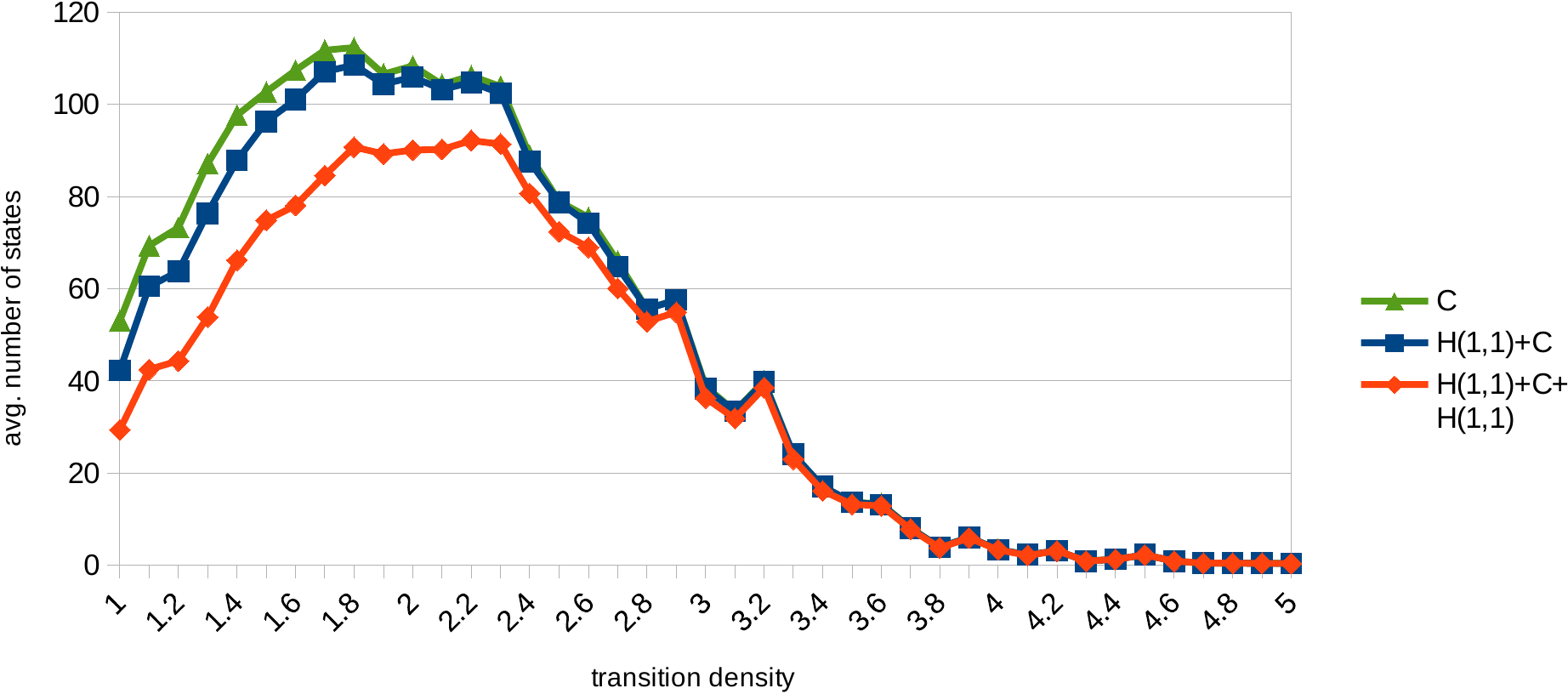} 
% \vspace{0.3cm}
\includegraphics[height=6.5cm]{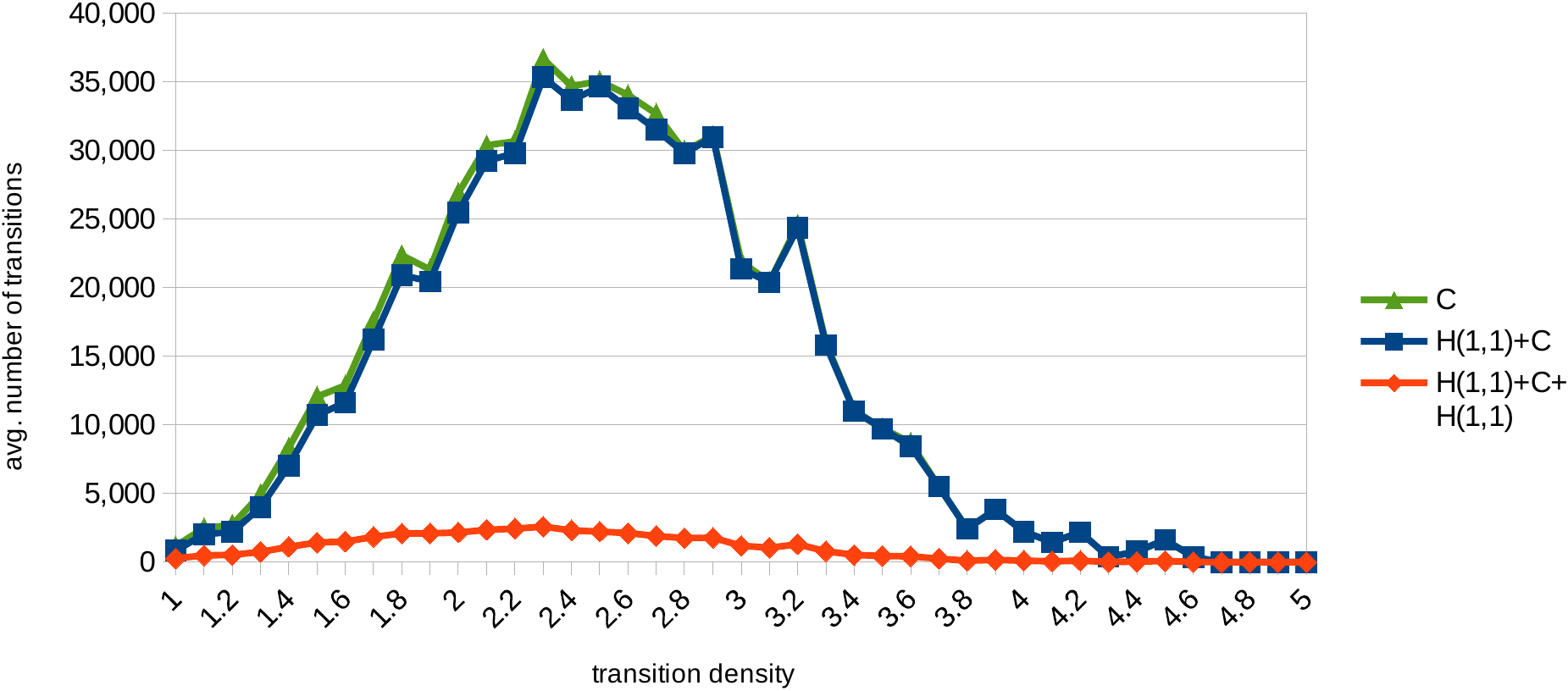} 
% \vspace{0.3cm}
\includegraphics[height=6.5cm]{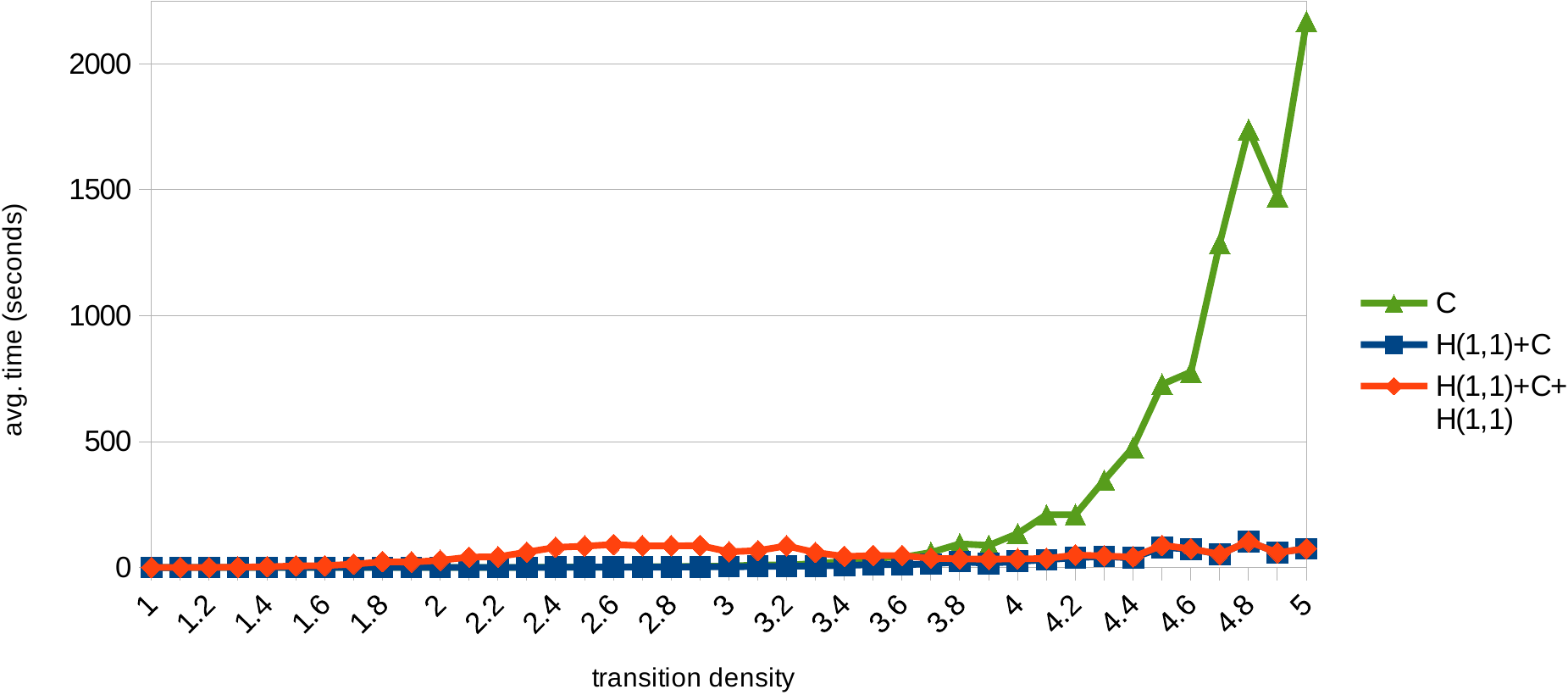} 
\end{center}
\caption{Reducing and complementing Tabakov-Vardi random tree automata with 7 states. 
Each data point is the average of $300$ automata.
In general, applying Heavy(1,1) before the complementation (H(1,1)+C)
yields smaller automata than direct complementation (C), on average.
When Heavy is used both before and after the complementation (H(1,1)+C+H(1,1)),
the difference is even more significant:
the automata produced by H(1,1)+C+H(1,1) had between 4 and 24 times less transitions than
those yielded by C,
but the greater reductions took longer to compute.
C still took the longest times recorded, for highly dense automata.}
\label{fig:chart_compl_q=7_app}
\end{figure}

\end{document}